\documentclass[onecolumn,10pt]{IEEEtran}
\usepackage{amsfonts,amsmath,amssymb,amsthm,mathrsfs}
\usepackage{graphicx,multirow,bm}
\usepackage{color}
\makeatletter
\newtheoremstyle{mythm}{3pt}{3pt}{}{16pt}{\bfseries}{:}{.5em}{}
\theoremstyle{mythm}
\newtheorem{theorem}{Theorem}
\newtheorem{definition}{Definition}
\newtheorem{remark}{Remark}

\newtheorem{corollary}{Corollary}
\newtheorem{lemma}{Lemma}

\newtheorem{construction}{Construction}

\newcommand{\cB}{\mathcal{B}}
\newcommand{\cC}{\mathcal{C}}

\newcommand{\cR}{\mathcal{R}}
\newcommand{\cS}{\mathcal{S}}
\newcommand{\cT}{\mathcal{T}}

\newcommand{\cV}{\mathcal{V}}

\newcommand{\cX}{\mathcal{X}}

\newcommand{\N}{\mathbb{N}}

\newcommand{\F}{\mathbb{F}}

\DeclareMathOperator{\loc}{Loc}
\DeclareMathOperator{\spn}{Span}
\DeclareMathOperator{\rank}{Rank}

\renewcommand{\le}{\leqslant}
\renewcommand{\leq}{\leqslant}
\renewcommand{\ge}{\geqslant}
\renewcommand{\geq}{\geqslant}

\newcommand{\mathset}[1]{\left\{#1\right\}}
\newcommand{\abs}[1]{\left|#1\right|}
\newcommand{\ceilenv}[1]{\left\lceil #1 \right\rceil}
\newcommand{\floorenv}[1]{\left\lfloor #1 \right\rfloor}

\newcommand{\parenv}[1]{\left( #1 \right)}

\newcommand{\oB}{\overline{B}}
\newcommand{\ocB}{\overline{\cB}}

\begin{document}
\title{On Optimal Locally Repairable Codes with Super-Linear Length\author{Han Cai, Ying Miao, Moshe Schwartz, and Xiaohu Tang}
\thanks{The material in this paper was submitted in part to the IEEE International Symposium on Information Theory 2019.}
\thanks{H. Cai and M. Schwartz are with the Department of Electrical and
Computer Engineering, Ben-Gurion University of the Negev, Beer Sheva
8410501, Israel (e-mail: hancai@aliyun.com; schwartz@ee.bgu.ac.il).}
\thanks{Y. Miao is with the Faculty of Engineering, Information and Systems,
University of Tsukuba, Tennodai 1-1-1, Tsukuba 305-8573, Japan (e-mail: miao@sk.tsukuba.ac.jp).}
\thanks{X. Tang is with the School of Information Science and Technology,
Southwest Jiaotong University, Chengdu, 610031, China (e-mail: xhutang@swjtu.edu.cn).}
  }
\date{}
\maketitle
\begin{abstract}
  Locally repairable codes which are optimal with respect to the bound
  presented by Prakash \emph{et al.} are considered. New upper bounds
  on the length of such optimal codes are derived. The new bounds both
  improve and generalize previously known bounds. Optimal codes are
  constructed, whose length is order-optimal when compared with the
  new upper bounds. The length of the codes is super-linear in the
  alphabet size.
\end{abstract}


\section{Introduction}

Large-scale cloud storage and distributed file systems, such as Amazon
Elastic Block Store (EBS) and Google File System (GoogleFS), have
reached such a massive scale that disk failures are the norm and not
the exception. In those systems, to protect the data from disk
failures, the simplest solution is a straightforward replication of
data packets across different disks.  However, this solution suffers
from a large storage overhead. As an alternative solution, $[n,k]$ MDS
codes are used as storage codes, which encode $k$ information symbols
to $n$ symbols and store them across $n$ disks. Using MDS codes
leads to a dramatic improvement in redundancy compared with
replication. However, for MDS codes, when one node fails, the system
recovers it at the cost of contacting $k$ surviving symbols, thus
complicating the repair process.

To improve the repair efficiently, in \cite{HCL}, locally repairable
codes were introduced to reduce the number of symbols contacted during
the repair process of a failed node.  More precisely, locally
repairable codes ensure that a failed symbol can be recovered by
accessing only $r\ll k$ other symbols \cite{HCL}.

The original concept of locality only works when exactly one erasure
occurs (that is, one node fails). Over the past few years, several
generalizations have been suggested for the definition of locality. As
examples we mention locality with a single repair set tolerating
multiple erasures \cite{PKLK}, locality with disjoint multiple
repairable sets \cite{WZ,RPDV,SES,CMST}, hierarchical locality \cite{SAK},
and unequal locality \cite{KL}.
For constructions of locally repairable codes with multiple
or uniform repair sets, refer to \cite{PHO,HX,BT,CCFT} as examples.

In this paper, we focus on locally repairable codes with a single
repair set that can repair multiple erasures locally \cite{PKLK}. By
ensuring $\delta-1\geq 2$ redundancies in each repair set, this kind of
locally repairable codes guarantees the system can recover from
$\delta-1$ erasures by accessing $r$ surviving code symbols for each erasure. This is
denoted as $(r,\delta)$-locality.

Research on codes with $(r,\delta)$-locality has proceeded along two
main tracks. In the first track, upper bounds on the minimum Hamming
distance and the code length have been studied. Singleton-type bounds
were introduced for codes with $(r,\delta)$-locality in
\cite{PKLK,SDYL,WZ15}. In \cite{CM}, a bound depending on the size of
the alphabet was derived for the Hamming distance of codes with
$(r,\delta)$-locality. Via linear programming, another bound related
with the size of the alphabet was introduced in \cite{ABHMT}. Very
recently, in \cite{GXY}, an interesting connection between the length
of optimal linear codes with $(r,2)$-locality and the size of the
alphabet was derived.

In the second research track, constructions for optimal locally
repairable codes have been studied. In \cite{RKSV}, a construction of
optimal locally repairable codes was introduced based on Gabidulin
codes  over a finite filed with size $q=\Theta((r+\delta-1)^{(rn)/(r+\delta-1)})$.
By analyzing the structure of repair sets, optimal locally repairable
codes were also constructed in \cite{SDYL} with $q=\Theta({n\choose k})$.
In \cite{TB}, a construction of optimal locally repairable
codes with $q=\Theta(n)$ was proposed. In \cite{TPD} and \cite{WFEH},
optimal locally repairable
codes were constructed using matroid theory. The construction of
\cite{TB} was generalized in \cite{KBTY} to include more flexible
parameters when $n\leq q$. Recently, in \cite{LXY}, cyclic
optimal locally repairable codes with unbounded length were
constructed for $\delta=2$ and Hamming distance $d=3,4$. Finally, for the case of
$\delta=2$ and Hamming distance $d=5$, \cite{GXY,Jin,BCGLP} presented constructions of locally
repairable codes that have optimal distance as well as
order-optimal length $n=\Theta(q^2)$.

The main contribution of this paper is the study of optimal linear
codes with $(r,\delta)$-locality and length that is super-linear in
the field size. We analyze the structure of optimal locally repairable
codes. As a result, we derive a new upper bound on the length of optimal locally repairable
codes for the case of $\delta>2$.
Secondly, as a byproduct, we prove that the bound for $\delta=2$ in
\cite{GXY} not only holds for some other cases (see Corollary \ref{corollary_bound_2} in this paper) besides the one mentioned in
\cite{GXY} but also can be improved for the case $d>r+\delta$.
  Finally, we give a general
construction of locally repairable codes with length that is
super-linear in the field size.  Based on some special structures such
as packings and Steiner systems, locally repairable codes with optimal
Hamming distances and order-optimal length $\Omega(q^{\delta})$ with respect to the new
bound $(\delta>2)$ are obtained. This is to say, the bound for
$\delta>2$ is also asymptotically tight for some special cases.

The remainder of this paper is organized as follows. Section
\ref{sec-preliminaries} introduces some preliminaries about locally
repairable codes.  Section \ref{sec-bound} establishes an upper bound
for the length of optimal locally repairable codes for the case
$\delta>2$.  Section \ref{sec-construction} presents a construction of
optimal locally repairable codes with length $n>q$.  Section
\ref{sec-conclusion} concludes this paper with some remarks.

\section{Preliminaries}\label{sec-preliminaries}

We present the notation and basic definitions used throughout the
paper. For a positive integer $n\in\N$, we define
$[n]=\{1,2,\dots,n\}$. For any prime power $q$, let $\F_q$
denote the finite field with $q$ elements. An $[n,k]_q$ linear code
$\cC$ over $\F_q$ is a $k$-dimensional subspace of
$\F_q^n$ with a $k\times n$ generator matrix $G=({\bf
  g}_1,{\bf g}_2,\dots,{\bf g}_{n})$, where ${\bf g}_i$ is a column
vector of dimension $k$ for all $i\in [n]$. Specifically, it is called
an $[n,k,d]_q$ linear code if the minimum Hamming distance is $d$. For
a subset $S\subseteq [n]$, let $|S|$ denote the cardinality of $S$,
let $2^S$ denote the set of all subsets of $S$, and define
\[ \rank(S)= \rank(\spn{\{{\bf g}_i |i\in  S\}}).\]

In \cite{GHSY}, Gopalan \emph{et al.} introduce the following definition
for the locality of code symbols.  The $i$th $(1 \leq i \leq n)$ code
symbol $c_i$ of an $[n, k,d]_q$ linear code $\cC$ is said to
have locality $r$ $(1 \leq r \leq k)$, if it can be recovered by
accessing at most $r$ other symbols in $\cC$. More precisely,
symbol locality can also be rigorously defined as follows.

\begin{definition}[\cite{GHSY}]\label{def_r_local}
For any column ${\bf g}_i$ of $G$ with $i\in [n]$, define $\loc({\bf
  g}_i)$ as the smallest integer $r$ such that there exists an
$(r+1)$-subset $R_i=\{i,i_1,i_2,\dots,i_r\}\subseteq [n]$ satisfying
\begin{equation}\label{eqn_def_locality}
{\bf g}_i \in \spn(R_i\setminus\{i\}),\,\, {\text{i.e.}},\,\, {\bf g}_i=\sum_{t=1}^{r}\lambda_t{\bf g}_{i_t},\ \ \ \ \lambda_t\in \F_q.
\end{equation}
Equivalently, for any codeword $C=(c_1,c_2,\dots,c_n)\in
\cC$, the $i$th component
$$c_i=\sum_{t=1}^{r}\lambda_tc_{i_t},\,\,\,\,\lambda_t\in
\F_q.$$ Define $\loc(S)=\max_{ i\in S}\loc({\bf g}_i)$ for any
set $S\subseteq[n]$.  Then, an $[n,k,d]_q$ linear code $\cC$
is said to have information locality $r$ if there exists $S\subseteq
[n]$ with $\rank(S)=k$ satisfying $\loc(S)=r.$ Furthermore, an
$[n,k,d]_q$ linear code $\cC$ is said to have all symbol locality r if
$\loc([n]) = r$.
\end{definition}

To guarantee that the system can locally recover from multiple
erasures, say, $\delta-1$ erasures, the definition of locality was
generalized in \cite{PKLK} as follows.

\begin{definition}[\cite{PKLK}]\label{def_r_delta_i} The $j$th column ${\bf g}_j$, $j\in [n]$, of a generator matrix $G$ of an $[n,k]_q$ linear code $\cC$ is said to have  $(r, \delta)$-locality  if
there exists a subset $S_j\subseteq [n]$ such that:
\begin{itemize}
  \item $j\in S_j$ and $|S_j|\leq r+\delta-1$; and
  \item the minimum Hamming distance of the punctured code $\cC|_{S_j}$ obtained by deleting the code symbols $c_t$ ($t \in [n]\setminus S_j$) is at least $\delta$,
\end{itemize}
where the set $S_j$ is also called a $(r,\delta)$-repair set of ${\bf g}_j$.  The
code $\cC$ is said to have information $(r,\delta)$-locality
if there exists $S\subseteq [n]$ with $\rank(S)=k$ such that for each
$j\in S$, ${\bf g}_j$ has $(r, \delta)$-locality. Furthermore, the
code $\cC$ is said to have all symbol $(r,\delta)$-locality if
all the code symbols have $(r,\delta)$-locality.
\end{definition}

In \cite{PKLK} (for the case $\delta=2$ \cite{GHSY}), the following upper bound on the minimum Hamming distance of linear codes with information $(r,\delta)$-locality was derived.
\begin{lemma}[\cite{PKLK}] \label{lemma_bound_i}
  For an $[n,k,d]_q$ linear code with information $(r,\delta)$-locality,
\begin{equation}\label{eqn_bound_for_local_i}
d\leq n-k+1-\left(\left\lceil\frac{k}{r}\right\rceil-1\right)(\delta-1).
\end{equation}
Additionally, a locally repairable code is said to be \emph{optimal}
if its minimum Hamming distance attains this bound with equality.
\end{lemma}

The following lemma is very useful to determine the minimum Hamming distance.
\begin{lemma}(\cite{MS})\label{lemma_rank_and_dist}
An $[n,k]_q$ linear code $\mathcal{C}$ has minimum Hamming distance
$d$ if and only if $d$ is the largest integer such that
\begin{equation*}
|S|\leq n-d
\end{equation*}
for every $S\subseteq [n]$ with ${\rm Rank}(S)\leq k-1$.
\end{lemma}

\section{Bounds on the Length of Locally Repairable Codes}\label{sec-bound}

The goal of this section is to derive upper bounds on the length of
optimal locally repairable codes. Throughout this section, let
\[n=(r+\delta-1)w+m, \qquad k=ru+v,\]
where $\delta\geq 2$, $0\leq m\leq r+\delta-2$, and $0\leq v\leq r-1$
are all integers.

For the bounds and the construction we shall require a simple
combinatorial covering design which we now define.

\begin{definition}
  Let $n,T,s\in\N$. Also, let $\cX$ be a set of cardinality $n$,
  whose elements are called \emph{points}. Finally, let
  $\cB=\{B_1,B_2,\dots,B_{T}\}\subseteq 2^\cX$ be a set of
  \emph{blocks} such that $\bigcup_{i\in[T]} B_i=\cX$, and for all
  $i\in[T]$, $\abs{B_i}\leq s$ and $\bigcup_{j\in
    T\setminus\mathset{i}}B_j\neq \cX$.  We then say $(\cX,\cB)$ is an
  \emph{$(n,T,s)$-essential covering family (ECF)}. If all blocks are
  the same size we say $(\cX,\cB)$ is a \emph{uniform} $(n,T,s)$-ECF.
\end{definition}

An important quantity associated with any family of subsets,
$\cB\subseteq 2^{\cX}$, is its \emph{overlap},
denoted $D(\cB)$, and defined as
\begin{equation*}
D(\cB)=\sum_{B\in \cB}|B|-\left|\bigcup\limits_{B\in\cB}B\right|.
\end{equation*}
Obviously $D(\cB)\geq 0$ and $D(\mathcal{B})$ is monotonically increasing. Additionally, $D(\cB)=0$ if and only if its
sets are pairwise disjoint.

Particularly, we need to investigate the structures of
repair sets in three lemmas, whose proofs are given in Appendix.

\begin{lemma}\label{lemma_for_D(B)}
  Let $(\cX,\cB)$ be an $(n,T,r+\delta-1)$-ECF, and assume it is
  non-uniform or that $D(\cB)\neq 0$. Then for every $0\leq t\leq T$,
  there exists a subset $\cB'\subseteq \cB$, $\abs{\cB'}=t$, such that
  \begin{equation*}
    t(r+\delta-1)-\left|\bigcup_{B\in \cB'}B\right|\geq\min\left\{r+\delta-1-m,\left\lfloor t/2\right\rfloor\right\}.
\end{equation*}
 \end{lemma}

\begin{lemma}\label{lemma_initialization}
For any $[n,k]_q$ linear code $\cC$ with all symbol
$(r,\delta)$-locality, let $\Gamma\subseteq 2^{[n]}$ be the set of all
possible $(r,\delta)$-repair sets. Then we can find a subset
$\cR\subseteq \Gamma$ such that $([n],\cR)$ is an
$(n,\abs{\cR},r+\delta-1)$-ECF with $|\mathcal{R}|\geq \lceil\frac{k}{r}\rceil$.
\end{lemma}

\begin{lemma}\label{lemma_find_V_1}
  Let $\cC$ be an $[n,k]_q$ linear code with all symbol
  $(r,\delta)$-locality. Let $\cR$ be the ECF given by Lemma
  \ref{lemma_initialization}. Assume $\cV\subseteq \cR$ such that
  $|\cV|\leq \lceil\frac{k}{r}\rceil-1$. If $\Delta$ is an integer
  such that
  \begin{equation}\label{eqn_delta}
    |\cV|(r+\delta-1)-\left|\bigcup_{R\in \cV}R\right|\geq \Delta>0
  \end{equation}
  and $\lceil\frac{k+\Delta}{r}\rceil>\lceil \frac{k}{r}\rceil$, then
  there exists a set $S\subseteq [n]$ with $\rank(S)=k-1$ and
  \begin{equation}\label{eqn_Size_S}
    |S|\geq
    k+\left(\left\lceil\frac{k}{r}\right\rceil-1\right)(\delta-1).
  \end{equation}
\end{lemma}

We are now at a position to state and prove the first main
tool in proving our bounds.

\begin{theorem}\label{theorem_repair_sets}
  Let $\cC$ be an optimal $[n,k,d]_q$ linear code with all symbol
  $(r,\delta)$-locality, where optimality is with respect to the bound
  in Lemma \ref{lemma_bound_i}. Let $\Gamma\subseteq 2^{[n]}$ be
  the set of all possible $(r,\delta)$-repair sets. Write $k=ru+v$,
  for integers $u$ and $v$, and $0\leq v\leq r-1$.  If $(r+\delta-1) |
  n$, $k>r$, and additionally, $u\geq 2(r-v+1)$ or $v=0$, then
  there exists a set of $(r,\delta)$-repair sets $\cS\subseteq
  \Gamma$, such that all $R\in\cS$ are of cardinality
  $\abs{R}=r+\delta-1$, and $\cS$ is a partition of $[n]$.
\end{theorem}
\begin{proof}
  Let $\cR\subseteq\Gamma$ be the ECF obtained in Lemma
  \ref{lemma_initialization}. If $D(\cR)=0$ and $\abs{R}=r+\delta-1$
  for all $R\in \cR$, then set $\cS=\cR$ the theorem follows.

  Otherwise, we have $D(\cR)\ne 0$ or $|R|<r+\delta-1$ for some $R\in
  \cR$. We distinguish between two cases. First, assume
  $k>2r$. By Lemma \ref{lemma_initialization},
  we know that $|\cR|\geq \lceil k/r\rceil$.
  According to Lemma \ref{lemma_for_D(B)} we can find a
  $(\lceil\frac{k}{r}\rceil-1)$-subset $\cV\subseteq\cR$ satisfying
\begin{equation*}
|\cV|(r+\delta-1)-\left|\bigcup_{R\in \cV}R\right|\geq \Delta=\min \left\{r+\delta-1,\left\lfloor\frac{\lceil\frac{k}{r}\rceil-1}{2}\right\rfloor\right\}
>0.
\end{equation*}
Since $u\geq 2(r-v+1)$ or $v=0$, we have
$\lceil\frac{k+\Delta}{r}\rceil>\lceil\frac{k}{r}\rceil$.  Therefore,
by Lemma \ref{lemma_find_V_1}, there is a set $S\subseteq[n]$ with
$\rank(S)=k-1$ and \begin{equation*} |S|\geq
  k+\left(\left\lceil\frac{k}{r}\right\rceil-1\right)(\delta-1).
\end{equation*}
Thus, by Lemma \ref{lemma_rank_and_dist}
\[ d \leq n-\abs{S} \leq n-k-\left(\left\lceil\frac{k}{r}\right\rceil-1\right)(\delta-1).\]
This is a contradiction to the optimality of $\cC$ with respect to the
bound in Lemma \ref{lemma_bound_i}.

In the second case, $r<k\leq 2r$. We note that we only need to
consider the case $v=0$, namely, $k=2r$, since if $v\ne 0$ then the
condition $u\geq 2(r-v+1)\geq 2$ implies that $k=ur+v>2r$.  We
therefore assume $k=2r$. If $D(\cR)\ne0$ or $|R|<r+\delta-1$ for some
$R\in \cR$ then we can find two distinct repair sets $R,R'\in \cR$
such that $R\cap R'\ne \emptyset$ or $\min(|R|,|R'|)<r+\delta-1$. In
either case, we have $\rank(R\cup R')<2r=k$.

We again distinguish between two cases depending on $|R\cap R'|$. For
the first case, if $|R\cap R'|\leq \min(|R|,|R'|)-\delta+1$ then we
have $\rank(R\cup R')\leq |R\cup R'|-2(\delta-1)<|R\cup R'|-\delta+1$.
In the second case, when $|R\cap R'|> \min(|R|,|R'|)-\delta+1$, assume
without loss of generality, that $|R\cap R'|> |R'|-\delta+1$, then
$\rank(R\cup R')=\rank(R)\leq |R|-\delta+1<|R\cup R'|-\delta+1.$

We now construct a set $S\subseteq[n]$ by arbitrarily adding
coordinates to $R\cup R'\subseteq S$ such that
$\rank(S)=k-1$. Therefore, $|S|-(k-1)\geq |R\cup R'|-\rank(R\cup
R')>\delta-1$, or equivalently, $\abs{S}\geq k+\delta-1$. Again by Lemma \ref{lemma_rank_and_dist},
we get
\[ d\leq n-\abs{S}\leq n-k-(\delta-1),\]
which is again a contradiction with the optimality of $\cC$ with
respect to the bound in Lemma \ref{lemma_bound_i}.
\end{proof}

We take a short break to consider the special case of $\delta=2$. This
special case was studied in \cite{GXY} and an upper bound on the
length of optimal codes was obtained.
\begin{theorem}[\cite{GXY}]\label{theorem_bound_old}
Let $\cC$ be an optimal $[n,k,d]_q$ code with all symbol $(r,2)$-locality.
If $d\geq 5$, $(r+1)|n$, and
$$\frac{n}{r+1}\geq \left(d-2-\left\lfloor\frac{d-2}{r+1}\right\rfloor\right)(3r+2)+\left\lfloor\frac{d-2}{r+1}\right\rfloor+1,$$
then
\begin{equation*}
\begin{split}
n\leq &\begin{cases}
\frac{(d-a)(r+1)}{4(q-1)r} q^{\frac{4(d-2)}{d-a}},& \text{if }a=1,2,\\
\frac{r+1}{r}\left(\frac{d-a}{4(q-1)} q^{\frac{4(d-3)}{d-a}}+1\right),& \text{if }a=3,4,\\
\end{cases}\\
=&\begin{cases}
O\parenv{dq^{\frac{4(d-2)}{d-a}-1}},& \text{if }a=1,2,\\
O\parenv{dq^{\frac{4(d-3)}{d-a}-1}},& \text{if }a=3,4,\\
\end{cases}
\end{split}
\end{equation*}
where $a\equiv d\pmod{4}$.
\end{theorem}

While we obtain the exact same
bound as \cite{GXY}, our bound is an improvement since it has more
relaxed conditions. In particular, the bound of Theorem \ref{theorem_bound_old} requires
$\frac{n}{r+1}\geq (d-2-\left\lfloor\frac{d-2}{r+1}\right\rfloor)(3r+2)+\lfloor\frac{d-2}{r+1}\rfloor+1,$
i.e., $k=\Omega(dr^2)$ whereas we require $k=\Omega(r^2)$. We now provide
the exact claim:

\begin{corollary}\label{corollary_bound_2}
Let $\cC$ be an optimal $[n,k,d]_q$ code with all symbol $(r,2)$-locality.
If $d\geq 5$, $k>r$, $(r+1)|n$, and additionally, $r|k$ or $u\geq 2(r+1-v)$
 (equivalently, $k\geq 2r^2+2r-(2r-1)\langle k \rangle_r$), then
\begin{equation*}
\begin{split}
n\leq &\begin{cases}
\frac{(d-a)(r+1)}{4(q-1)r} q^{\frac{4(d-2)}{d-a}},& \text{if }a=1,2,\\
\frac{r+1}{r}\left(\frac{d-a}{4(q-1)} q^{\frac{4(d-3)}{d-a}}+1\right),& \text{if }a=3,4,\\
\end{cases}\\
=&\begin{cases}
O\parenv{dq^{\frac{4(d-2)}{d-a}-1}},& \text{if }a=1,2,\\
O\parenv{dq^{\frac{4(d-3)}{d-a}-1}},& \text{if }a=3,4,\\
\end{cases}
\end{split}
\end{equation*}
where $a\equiv d\pmod{4}$ and $\langle k \rangle_r$ denotes the least nonnegative integer congruent to $k$ modulo $r$.
\end{corollary}
\begin{proof}
The desired result directly follows by
  replacing \cite[Theorem 3.1]{GXY} with Theorem
  \ref{theorem_repair_sets}, and continuing with the same proof as
  \cite[Theorem 3.2]{GXY}.
\end{proof}

We bring another corollary that stems from Theorem
\ref{theorem_repair_sets}.  It slightly extends \cite[Theorem 9]{SDYL}, originally proved only for $r|k$, and has a very similar
proof which we give for completeness.

\begin{corollary}\label{corollary_delta>2}
Let $\cC$ be an optimal $[n,k,d]_q$ linear code with all symbol
$(r,\delta)$-locality, where optimality is with respect to the bound
in Lemma \ref{lemma_bound_i}. If $k>r$, $n=w(r+\delta-1)$, and
additionally $r|k$ or $u\geq 2(r+1-v)$, then there are
$w$ pairwise-disjoint $(r,\delta)$-repair sets,
$R_1,\dots,R_w\subseteq[n]$, such that for all $1\leq i\leq w$,
$\abs{R_i}=r+\delta-1$, and the punctured code $\cC|_{R_i}$ is a
linear $[r+\delta-1,r,\delta]_q$ MDS code.
\end{corollary}
\begin{proof}
  We contend that the repair sets, $\cS$, from Theorem
  \ref{theorem_repair_sets}, satisfy the requirements. Thus, it
  remains to prove that for each $\cC|_{R}$, $R\in\cS$, the Hamming
  distance is exactly $\delta$.
  Assume to the contrary, and without
  loss of generality, that $d(\cC|_{R_1})>\delta$.

  Note that $\bigcup_{1\leq i\leq w}R_i=[n]$
  means $\rank(\bigcup_{1\leq i\leq w}R_i)=k$ and then
  $w=\frac{n}{r+\delta-1}\geq \lceil\frac{k}{r}\rceil$ since
  $\rank(R_i)\leq r$ for $1\leq i\leq w$. Also recall our notation that $v\equiv k\bmod
  r$ and $0\leq v<r$.
  Fix some arbitrary set $R'\subseteq R_{\lceil\frac{k}{r}\rceil}$,
  with $\abs{R'}=v$ if $v\ne 0$, and $\abs{R'}=r$ if
  $v=0$. Consider now the set
  $$S=R'\cup\parenv{\bigcup_{1\leq i\leq \lceil\frac{k}{r}\rceil-1}R_{i}}.$$
  By the Singleton bound we have,
\begin{equation*}
\begin{split}
\rank(S)&\leq\rank(R')+\sum_{1\leq i\leq \lceil\frac{k}{r}\rceil-1}\rank(R_{i})\\
&\leq \begin{cases}
v+\sum_{1\leq i\leq \lceil\frac{k}{r}\rceil-1}(r+\delta-1-d(\cC|_{R_{i}})+1)< v+r(\lceil\frac{k}{r}\rceil-1)=k, & \text{ if $v\ne 0$,}\\
r+\sum_{1\leq i\leq \lceil\frac{k}{r}\rceil-1}(r+\delta-1-d(\cC|_{R_{i}})+1)< r+r(\lceil\frac{k}{r}\rceil-1)=k, & \text{ if $v=0$.}
\end{cases}
\end{split}
\end{equation*}
We also have
\begin{align*}
  |S| &= \begin{cases}
    v+(r+\delta-1)\parenv{\ceilenv{\frac{k}{r}}-1}, & \text{if $v\ne 0$,}\\
    r+(r+\delta-1)\parenv{\ceilenv{\frac{k}{r}}-1}, & \text{if $v=0$,}
  \end{cases}\\
    &=k+\parenv{\ceilenv{\frac{k}{r}}-1}(\delta-1).
\end{align*}
But now this contradicts the optimality of $\cC$ by Lemma \ref{lemma_rank_and_dist}.
\end{proof}

We now extend our scope and consider locally repairable codes for the
case of $\delta>2$. In the sequel, the discussion is based on the
structure of the repair sets given in Corollary \ref{corollary_delta>2}.

\begin{lemma}\label{lemma_bound_delta>2}
Let $n=w(r+\delta-1)$, $\delta>2$, $k=ur+v>r$, and additionally,
$r|k$ or $u\geq 2(r+1-v)$, where all parameters are integers. If there
exists an optimal $[n,k,d]_q$ linear code $\cC$ with all symbol
$(r,\delta)$-locality, then there exists a $[w(r+1),k,d']_q$ linear
code $\cC'$ with all symbol $(r,2)$-locality (i.e., locality $r$),
and $d'\geq 2\floorenv{(d-1)/\delta}+1$.
\end{lemma}
\begin{proof}
By Corollary \ref{corollary_delta>2}, and up to a rearrangement of the
code coordinates, the code $\cC$ has parity-check matrix $P$ of the
following form,
\begin{equation*}
P=\begin{pmatrix}
      L^{(1)} & 0 & 0 & \dots & 0 \\
      0 &  L^{(2)} & 0 & \dots & 0 \\
      0 & 0 & L^{(3)} & \dots & 0 \\
      \vdots & \vdots & \vdots & \ddots & \vdots \\
      0 & 0 & 0 & 0 & L^{(w)} \\
      H_1 & H_2 & H_3 & \dots & H_w \\
    \end{pmatrix},
\end{equation*}
where $L^{(i)}=(I_{\delta-1}, P_i)$ is a $(\delta-1)\times
(r+\delta-1)$ matrix for all $1\leq i\leq w$. Herein, without loss of generality, we
assume $L^{(i)}$ with canonical form for $1\leq i\leq w$. For all $1\leq i\leq w$, rewrite the
$(\delta-1)\times (r+\delta-1)$ matrix $L^{(i)}=(I_{\delta-1} P_i)$ as
\[L^{(i)}=
\begin{pmatrix}
  L^{(i)}_{1,1} & L^{(i)}_{1,2} \\
  L^{(i)}_{2,1} & L^{(i)}_{2,2} \\
\end{pmatrix},\]
where $L^{(i)}_{2,2}$ is a $(\delta-2)\times (\delta-2)$ matrix. It is
easy to check that $\det(L^{(i)}_{2,2})\ne 0$ for all $1\leq i\leq w$,
since $L^{(i)}$ is a parity-check matrix of an
$[r+\delta-1,r,\delta]_q$ MDS code according to Corollary
\ref{corollary_delta>2}.  By column linear transformations, the matrix
$P$ is equivalent to
\begin{equation}\label{eqn_d_columns_P_trans}
\begin{pmatrix}
  Q_1 & 0 & 0 & \dots & 0 \\
  0 &  Q_2 & 0 & \dots & 0 \\
  0 & 0 & Q_3 & \dots & 0 \\
  \vdots & \vdots & \vdots & \ddots & \vdots \\
  0 & 0 & 0 & 0 & Q_w \\
  H'_1 & H'_2 & H'_3 & \dots & H'_w \\
\end{pmatrix},
\end{equation}
where
\begin{align}
\label{eqn_def_R}
Q_i&=\begin{pmatrix}
        Q_{i,1}=L^{(i)}_{1,1}-L^{(i)}_{1,2}(L^{(i)}_{2,2})^{-1}L^{(i)}_{2,1} & L^{(i)}_{1,2} \\
        0 & L^{(i)}_{2,2} \\
\end{pmatrix},\\
\label{eqn_def_H}
H'_i&=(H'_{i,1}=H_{i,1}-H_{i,2}(L^{(i)}_{2,2})^{-1}L^{(i)}_{2,1},H'_{i,2}=H_{i,2}) \text{ with } H_i=(H_{i,1},H_{i,2}).
\end{align}

Now consider the code $\cC'$ with parity-check matrix
\begin{equation}\label{eqn_matrix_p'}
P'=\begin{pmatrix}
Q_{1,1} & 0 & 0 & \dots & 0 \\
0 &  Q_{2,1} & 0 & \dots & 0 \\
0 & 0 & Q_{3,1} & \dots & 0 \\
\vdots & \vdots & \vdots & \ddots & \vdots \\
0 & 0 & 0 & 0 & Q_{w,1} \\
H_{1,1}' & H_{2,1}' & H_{3,1}' & \dots & H_{w,1}' \\
\end{pmatrix},
\end{equation}
where $Q_{i,1}$ and $H'_{i,1}$, for $1\leq i\leq w$, are defined by
\eqref{eqn_def_R} and \eqref{eqn_def_H}, respectively.

Given a set of coordinates
$T=\mathset{t_1,\dots,t_\ell}\subseteq[r+\delta-1]$, and given
$A=(A_1,\dots,A_{r+\delta-1})$, we define the projection of $A$
onto $T$ by $\Delta_{T}(A)=(A_{t_1},A_{t_2},\dots,A_{t_l})$ (where the
order of coordinates in the projection will not matter to us). We
emphasize that $Q_{i,1}$, for all $1\leq i\leq w$, does not have a
zero coordinate, since according to Corollary \ref{corollary_delta>2},
$\Delta_{S_\tau}(Q_i)$ has full rank, where we define
$S_\tau=\{\tau\}\cup\{r+2,r+3,\dots,r+\delta-1\}$, $\tau\in[r+1]$.  Thus,
by \eqref{eqn_matrix_p'}, $\cC'$ is a code with all symbol
$(r,2)$-locality.

To complete the proof we only need to show $d'\geq 2t+1$, where we
define $t=\floorenv{(d-1)/\delta}$. Namely, we need to show that any
$2t$ columns of $P'$ are linearly independent. A selection of $2t$
columns from $P'$, denoted by $\cT'$, has the following general form,
\begin{equation*}
\Delta_{\cT'}(P')\triangleq\begin{pmatrix}
      \Delta_{T'_1}(Q_{1,1}) & 0 & 0 & \dots & 0 \\
      0 &  \Delta_{T'_2}(Q_{2,1}) & 0 & \dots & 0 \\
      0 & 0 & \Delta_{T'_3}(Q_{3,1}) & \dots & 0 \\
      \vdots & \vdots & \vdots & \ddots & \vdots \\
      0 & 0 & 0 & 0 & \Delta_{T'_w}(Q_{w,1}) \\
     \Delta_{T'_1}(H'_{1,1}) & \Delta_{T'_2}(H'_{2,1}) & \Delta_{T'_3}(H'_{3,1}) & \dots & \Delta_{T'_w}(H'_{w,1}) \\
  \end{pmatrix},
\end{equation*}
where $\sum_{1\leq i\leq w}|T'_i|=2t$. Since the locality of $\cC'$
guarantees recovery from any one erasure independently, the
non-trivial cases to consider are those where $T'_{\tau_i}\geq 2$ for
$1\leq \tau_i\leq w$ and $1\leq i\leq s$, where $s$ denotes the number
of sets $T'_i$ with $|T'_i|\geq 2$ and $s\leq \min(t,w)$.

With a coordinate selection $\cT'$ from $P'$ we naturally associate a
coordinate selection $\cT$ from $P$, defined by
\[
T_{\tau_i}=
T'_{\tau_i}\cup\{r+2,r+2,\dots, r+\delta-1\},
\]
for $1\leq i\leq s$, and with $\sum_{1\leq i\leq
  s}|T_{\tau_i}|=2t+s(\delta-2)\leq t\delta\leq d-1$. Recall that if
$\{r+2,r+3,\dots,r+\delta-1\}\subset T\subseteq [r+\delta-1]$ then
\eqref{eqn_d_columns_P_trans}, \eqref{eqn_def_R} and \eqref{eqn_def_H}
imply that
\[\begin{pmatrix}
\Delta_{T}(L^{(i)})\\
\Delta_{T}(H_i) \\
\end{pmatrix}
\qquad\text{and}\qquad
\begin{pmatrix}
  \Delta_{T}(Q_i)\\
  \Delta_{T}(H'_i) \\
\end{pmatrix}
\]
are rank equivalent, based on only invertible column linear
transformations for $1\leq i\leq w$.  Note that the distance of $\cC$
satisfies $d\geq \delta t + 1\geq 2t+s(\delta-2)+1$, which implies
that any $\sum_{1\leq i\leq s}|T_{\tau_i}|\leq 2t+s(\delta-2)$ columns
of $P$ have full rank of $\sum_{1\leq i\leq s}|T_{\tau_i}|$, i.e.,
\begin{equation}\label{eqn_rank_T}
\begin{split}
\sum_{1\leq i\leq s}|T_{\tau_i}|=&\rank\begin{pmatrix}
\Delta_{T_{\tau_1}}(L^{(\tau_1)}) & 0 & 0 & \dots & 0 \\
0 &  \Delta_{T_{\tau_2}}(L^{(\tau_2)}) & 0 & \dots & 0 \\
0 & 0 & \Delta_{T_{\tau_3}}(L^{(\tau_3)}) & \dots & 0 \\
\vdots & \vdots & \vdots & \ddots & \vdots \\
0 & 0 & 0 & 0 & \Delta_{T_{\tau_s}}(L^{(\tau_s)}) \\
\Delta_{T_{\tau_1}}(H_{\tau_1}) & \Delta_{T_{\tau_2}}(H_{\tau_2}) & \Delta_{T_{\tau_3}}(H_{\tau_3}) & \dots & \Delta_{T_{\tau_s}}(H_{\tau_s}) \\
\end{pmatrix}\\
=&\rank\begin{pmatrix}
\Delta_{T_{\tau_1}}(Q_{\tau_1}) & 0 & 0 & \dots & 0 \\
0 &  \Delta_{T_{\tau_2}}(Q_{\tau_2}) & 0 & \dots & 0 \\
0 & 0 & \Delta_{T_{\tau_3}}(Q_{\tau_3}) & \dots & 0 \\
\vdots & \vdots & \vdots & \ddots & \vdots \\
0 & 0 & 0 & 0 & \Delta_{T_{\tau_s}}(Q_{\tau_s}) \\
\Delta_{T_{\tau_1}}(H'_{\tau_1}) & \Delta_{T_{\tau_2}}(H'_{\tau_2}) & \Delta_{T_{\tau_3}}(H'_{\tau_3}) & \dots & \Delta_{T_{\tau_s}}(H'_{\tau_s}) \\
\end{pmatrix},
\end{split}
\end{equation}
where the second equality holds by \eqref{eqn_d_columns_P_trans},
\eqref{eqn_def_R}, \eqref{eqn_def_H} and the fact that
$\{r+2,r+3,\dots,r+\delta-1\}\subseteq T_{\tau_i}$ for $1\leq i\leq
s$.  Therefore, by \eqref{eqn_def_R}, \eqref{eqn_def_H}, and
\eqref{eqn_rank_T}, we have
 \begin{equation*}\label{eqn_delta_p'}
\rank\left(\Delta_{\cT'}(P')\right)=
\rank\begin{pmatrix}
\Delta_{T'_{\tau_i}}(Q_{\tau_1,1}) & 0 & 0 & \dots & 0 \\
0 &  \Delta_{T'_{\tau_2}}(Q_{\tau_2,1}) & 0 & \dots & 0 \\
0 & 0 & \Delta_{T'_{\tau_3}}(Q'_{\tau_3,1}) & \dots & 0 \\
\vdots & \vdots & \vdots & \dots & \vdots \\
0 & 0 & 0 & 0 & \Delta_{T'_{\tau_s}}(Q_{\tau_s,1}) \\
\Delta_{T'_{\tau_1}}(H'_{\tau_1,1}) & \Delta_{T'_{\tau_2}}(H'_{\tau_2,1}) & \Delta_{T'_{\tau_3}}(H'_{\tau_3,1}) & \dots & \Delta_{T'_{\tau_s}}(H'_{\tau_s,1}) \\
\end{pmatrix}
=\sum_{1\leq i\leq s}|T'_{\tau_i}|,
\end{equation*}
where $T'_{\tau_i}=T_{\tau_i}\setminus\{r+2,r+3,\dots,r+\delta-1\}$
for $1\leq i\leq s$.  This is to say, the code $\cC'$ can recover from any
$2t$ erasures, hence, $d'\geq 2t+1$.
\end{proof}

The following bound is derived from Lemma
\ref{lemma_bound_delta>2}. The proof follows the same path as the
proof of \cite[Theorem 3.2]{GXY}. We bring it here for completeness.

\begin{theorem}\label{theorem_bound_delta>2}
Let $n=w(r+\delta-1)$, $\delta\geq 2$, $k=ur+v$, and additionally,
$r|k$ or $u\geq 2(r+1-v)$, where all parameters are integers. Assume there
exists an optimal $[n,k,d]_q$ linear code $\cC$ with all symbol
$(r,\delta)$-locality, and define $t=\floorenv{(d-1)/\delta}$.  If
$2t+1>4$, then
\begin{align*}
n&\leq
\begin{cases}
\frac{(t-1)(r+\delta-1)}{2r(q-1)}q^{\frac{2(w-u)r-2v}{t-1}}, &\text{ if } t \text{ is odd},  \\
\frac{t(r+\delta-1)}{2r(q-1)}q^{\frac{2(w-u)r-2v}{t}}, &\text{ if } t \text{ is even},\\
\end{cases}\\
&=O\parenv{\frac{t(r+\delta)}{r}q^{\frac{(w-u)r-v}{\floorenv{t/2}}-1}},
\end{align*}
where $w-u$ can also be rewritten as $w-u=\lfloor(d-1+v)/(r+\delta-1)\rfloor$.
\end{theorem}
\begin{proof}
By Lemma \ref{lemma_bound_delta>2}, we have a $[w(r+1),k,d_1\geq
  2t+1]_q$ linear code with all symbol $(r,2)$-locality and
parity-check matrix given by \eqref{eqn_matrix_p'}. Equivalently,
there is a $[w(r+1),k,d_1\geq 2t+1]_q$ linear code $\cC_1$ with all symbol
$(r,2)$-locality and parity-check matrix given by
\begin{equation}\label{eqn_matrix_p_1}
M_1=\begin{pmatrix}
J & 0 & 0 & \dots & 0 \\
0 &  J & 0 & \dots & 0 \\
0 & 0 & J & \dots & 0 \\
\vdots & \vdots & \vdots & \ddots & \vdots \\
0 & 0 & 0 & 0 & J \\
H^{(1)}_{1} & H^{(1)}_{2} & H^{(1)}_{3} & \dots & H^{(1)}_{w} \\
\end{pmatrix},
\end{equation}
where $J$ is the $1\times (r+1)$ all-ones matrix.

Let us denote the columns of $H^{(1)}_i$ by
$H^{(1)}_{i}=(h^{(1)}_{i,1},h^{(1)}_{i,2},\dots,
h^{(1)}_{i,r+1})$. Based on $M_1$ we can generate a matrix $M_2$ defined by
\begin{equation}\label{eqn_matrix_p_2}
M_2=
\begin{pmatrix}
  H^{(1)}_{1,2} & H^{(1)}_{2,2} & H^{(1)}_{3,2} & \dots & H^{(1)}_{w,2} \\
\end{pmatrix},
\end{equation}
where
$H^{(1)}_{i,2}=(h^{(1)}_{i,2}-h^{(1)}_{i,1},h^{(1)}_{i,3}-h^{(1)}_{i,1},
\dots, h^{(1)}_{i,r+1}-h^{(1)}_{i,1})$. Since $M_1$ is the
parity-check matrix of a $[w(r+1),k,d_1\geq 2t+1]_q$ linear code with
all symbol $(r,2)$-locality, by \eqref{eqn_matrix_p_1} and
\eqref{eqn_matrix_p_2}, we have that $M_2$ is the parity-check matrix
of a linear code $\cC_2$, with parameters $[wr,k=ur+v,d_2\geq t+1]_q$.

Now we apply the Hamming bound \cite{MS} to $\cC_2$. We distinguish
between two cases, depending on the parity of $t$.

Case 1: $t$ is odd.  In this case,  by the Hamming bound,
\begin{equation*}
q^{ur+v}\leq \frac{q^{wr}}{\sum_{0\leq i\leq \frac{t-1}{2}}\binom{wr}{i}(q-1)^i}\leq \frac{q^{wr}}{\binom{wr}{\frac{t-1}{2}}(q-1)^{\frac{t-1}{2}}}\leq \frac{q^{wr}}{\left(\frac{wr}{\frac{t-1}{2}}\right)^{\frac{t-1}{2}}(q-1)^{\frac{t-1}{2}}},
\end{equation*}
i.e.,
\begin{equation*}
wr\leq \frac{t-1}{2(q-1)}q^{\frac{2(w-u)r-2v}{t-1}}.
\end{equation*}
This is to say,
\[n\leq \frac{(r+\delta-1)(t-1)}{2r(q-1)}q^{\frac{2(w-u)r-2v}{t-1}}.\]

Case 2: $t$ is even.  Similarly,  by the Hamming bound, we have
\begin{equation*}
q^{ur}\leq \frac{q^{wr}}{\sum_{1\leq i\leq \frac{t}{2}}\binom{wr}{i}(q-1)^i}\leq \frac{q^{wr}}{\binom{wr}{\frac{t}{2}}(q-1)^{\frac{t}{2}}}\leq \frac{q^{wr}}{\left(\frac{wr}{\frac{t}{2}}\right)^{\frac{t}{2}}(q-1)^{\frac{t}{2}}},
\end{equation*}
which means
$$n\leq \frac{t(r+\delta-1)}{2r(q-1)}q^{\frac{2(w-u)r-2v}{t}}.$$

By Lemma \ref{lemma_bound_i}, $\mathcal{C}$ is optimal means that
\begin{equation*}
d-1=\begin{cases}
w(r+\delta-1)-ur-v-u(\delta-1), &\text{ if } v\ne 0,\\
w(r+\delta-1)-ur-(u-1)(\delta-1), &\text{ if } v= 0,\\
\end{cases}
\end{equation*}
i.e., $w-u=\lfloor(d-1+v)/(r+\delta-1)\rfloor$. This completes the proof.
\end{proof}

Recalling Corollary \ref{corollary_delta>2} again, we can improve the performance of
the bounds on the length of optimal locally repairable codes with all symbol
$(r,\delta)$-locality for the case $d>r+\delta$ by the following corollary.

\begin{corollary}\label{corollary_reduce_d}
Let $n=w(r+\delta-1)$, $\delta\geq 2$, $k=ur+v>r$, and additionally,
$r|k$ or $u\geq 2(r+1-v)$, where all parameters are integers. If there
exists an optimal $[n,k,d]_q$ linear code $\cC$ with $d>r+\delta$ and
all symbol $(r,\delta)$-locality, then there exists an optimal linear
code $\cC'$ with all symbol $(r,\delta)$-locality and parameters
$[n-\epsilon(r+\delta-1),k,d'=d-\epsilon(r+\delta-1)]_q,$ where
$\epsilon=\lceil (d-1)/(r+\delta-1)\rceil-1$.
\end{corollary}

\begin{proof}
By Corollary \ref{corollary_delta>2}, there are $R_1,R_2,\cdots,R_w$ such
that $\cC|_{R_i}$, $1\leq i\leq w$, is an $[r+\delta-1,r,\delta]_q$ MDS code.
Note that $\epsilon=\lceil (d-1)/(r+\delta-1)\rceil-1$. The fact $\cC$ is optimal means that
\begin{equation}\label{eqn_d_for_C}
d=n-k+1-\left(\left\lceil\frac{k}{r}\right\rceil-1\right)(\delta-1),
\end{equation}
by Lemma \ref{lemma_bound_i}.
Recall that $k>r$, $n=w(r+\delta-1)$, and $d>r+\delta$. Thus, we have
$1\leq \epsilon\leq w-1$. Now let $\cC'$ be the punctured code of $\cC$ over
the set $W=\bigcup_{\epsilon+1\leq i\leq w-1} R_i$, i.e., $\cC'=\cC|_{W}$.
The fact $\cC'|_{R_i}=\cC|_{R_i}$ for $\epsilon+1\leq i\leq w$ is an $[r+\delta-1,r,\delta]_q$
MDS code means that $\cC'$ has all symbol $(r,\delta)$-locality.
The fact $\cC'=\cC|_{W}$ implies $n'=n-\sum_{1\leq i\leq \epsilon}|R_i|=n-\epsilon(r+\delta-1)$
and
$$d'\geq d-\sum_{1\leq i\leq \epsilon}|R_i|=d-\epsilon(r+\delta-1).$$
However, by Lemma \ref{lemma_bound_i}, we have
\begin{equation*}
d'\leq n'-k+1-\left(\left\lceil\frac{k}{r}\right\rceil-1\right)(\delta-1)
=n-\epsilon(r+\delta-1)-k+1-\left(\left\lceil\frac{k}{r}\right\rceil-1\right)(\delta-1)
=d-\epsilon(r+\delta-1),
\end{equation*}
where the last equality follows by \eqref{eqn_d_for_C}.
Thus, we have $d'=d-\epsilon(r+\delta-1)$. Again by Lemma \ref{lemma_bound_i}
the code $\cC'$ is also an optimal linear code with all symbol
$(r,\delta)$-locality and parameters
$[n-\epsilon(r+\delta-1),k,d'=d-\epsilon(r+\delta-1)]_q,$
which completes the proof.
\end{proof}

By Corollary \ref{corollary_reduce_d}, we can firstly reduce the optimal
locally repairable code $\cC$ into an optimal locally repairable code
$\cC'$ with $d'\leq r+\delta$ and then apply Theorem \ref{theorem_bound_delta>2} ($\delta>2$)
and Corollary \ref{corollary_bound_2} ($\delta=2$) to get an upper bound for the length
of $\cC$.

\begin{corollary}
Let $n=w(r+\delta-1)$, $\delta\geq 2$, $k=ur+v>r$, and additionally,
$r|k$ or $u\geq 2(r+1-v)$, where all parameters are integers. If there
exists an optimal $[n,k,d]_q$ linear code $\cC$ with $d>r+\delta$ and
all symbol $(r,\delta)$-locality, then for $\delta=2$
\begin{equation*}
n\leq \epsilon(r+\delta-1)+
\begin{cases}
\frac{(d'-a)(r+1)}{4(q-1)r} q^{\frac{4(d'-2)}{d'-a}},& \text{if }a=1,2,\\
\frac{r+1}{r}\left(\frac{d'-a}{4(q-1)} q^{\frac{4(d'-3)}{d'-a}}+1\right),& \text{if }a=3,4,\\
\end{cases}
\end{equation*}
and for $\delta>2$
\begin{align*}
n&\leq \epsilon(r+\delta-1)+
\begin{cases}
\frac{(t-1)(r+\delta-1)}{2r(q-1)}q^{\frac{2(w'-u)r-2v}{t-1}}, &\text{ if } t \text{ is odd},  \\
\frac{t(r+\delta-1)}{2r(q-1)}q^{\frac{2(w'-u)r-2v}{t}}, &\text{ if } t \text{ is even},\\
\end{cases}\\
\end{align*}
where $\epsilon=\lceil(d-1)/(r+\delta-1)\rceil-1$, $d'=d-\epsilon(r+\delta-1)$,
$w'=w-\epsilon$, $a\equiv d'\pmod 4$, and $t=\lfloor(d'-1)/(\delta)\rfloor$ so that $2t+1>4$ holds.

\end{corollary}

In the next section, we will prove that the bound in Theorem
\ref{theorem_bound_delta>2} is asymptotically tight for some special
cases, i.e., there indeed exist some optimal linear codes with all
symbol $(r,\delta)$-locality and asymptotically optimal length. In addition,
we will also prove the condition $2t+1>4$ is necessary, by constructing linear codes with length
independent of the field size $q$ for the case $2t+1\leq 4$.

\section{Optimal Locally Repairable Codes with Super-Linear Length}\label{sec-construction}

In this section, we introduce a generic construction of locally
repairable codes. Next, we demonstrate applications by this construction by employing some combinatorial structures to generate  optimal locally
repairable codes with length $n$ that is super-linear in the field
size $q$.

\subsection{A general construction}
In the subsection, to streamline the presentation  we adopt a
slightly different notation than the previous one: we use
$n=w(r+\delta-1)$ and $k=(w-1)r+v$ for $0< v\leq r$, where all
parameters are integers.

\begin{construction}
  \label{cons}
  Let the $k$ information symbols be partitioned into $w$ sets, say,
  \begin{align*}
    I^{(i)}&=\{I_{i,1},I_{i,2},\dots,I_{i,r}\}, \quad \text{for $i\in[w-1]$,}\\
    I^{(w)}&=\{I_{w,1},I_{w,2},\dots,I_{w,v}\}.
  \end{align*}
  A linear code with length $n$ is constructed by describing a linear
   map from the information ${\bm I}=(I_{1,1},\dots, I_{w,v})\in
  \F^k_{q}$ to a codeword ${\bm C}({\bm
    I})=(c_{1,1},\dots,c_{w,r+\delta-1})\in \F^n_q$, thus the
  $[n,k]_q$ linear code is $\cC=\{{\bm C}({\bm I}) ~:~ {\bm I}\in
  \F_q^{k}\}$. This mapping is performed by the following three steps:

  \paragraph{Step 1 -- Partial parity check symbols}
  For $1\leq i\leq w-1$, let $S_i=\{\theta_{i,t} ~:~ 1\leq t\leq
  r+\delta-1\}\subseteq \F_q$ and let $f_i(x)$ be the unique
  polynomial over $\F_q$ with $\deg(f_i)\leq r-1$ that satisfies
  $f_i(\theta_{i,t})=I_{i,t}$ for $1\leq t\leq r$.  For $1\leq i\leq
  w-1$ and $1\leq t\leq r+\delta-1$, set $c_{i,t}=f_i(\theta_{i,t})$.

\paragraph{Step 2 -- Auxiliary symbols} Let $\{\alpha_t ~:~ 1\leq t\leq r-v\}\subseteq \F_q\setminus (\bigcup_{1\leq i\leq w-1}S_i)$. For $1\leq i\leq w-1$, and $1\leq t\leq r-v$, define
\begin{equation}\label{eqn_A_i}
a_{i,t}=\frac{f_i(\alpha_t)}{\prod_{\theta\in
    S_i}(\alpha_t-\theta)}.
\end{equation}

\paragraph{Step 3 -- Global parity check symbols}
Let $S_w=\{\theta_{w,t} ~:~ 1\leq t\leq r+\delta-1\}\subseteq
\F_q\setminus\{\alpha_t ~:~ 1\leq t\leq r-v\}$ and let $f_w(x)$ be the
unique polynomial over $\F_q$ with $\deg(f_w)\leq r-1$ that satisfies
$f_w(\theta_{w,t})=I_{w,t}$ for $1\leq t\leq v$, as well as
\begin{equation}\label{eqn_global_check}
\sum_{1\leq i\leq w}a_{i,t}=0 \text{ for }1\leq t\leq r-v,
\end{equation}
where $a_{w,t}=\frac{f_w(\alpha_t)}{\prod_{\theta\in
    S_w}(\alpha_t-\theta)}$ for $1\leq t\leq r-v$. Here, the
polynomial $f_w(x)$ can be viewed as a polynomial over $\F_q$
determined by $I_{w,j}$, $1\leq j\leq v$ and $a_{w,t}$ for $1\leq
t\leq r-v$. Thus, $f_w(x)$ is unique and well defined. Set
$c_{w,j}=f_w(\theta_{w,j})$, for $1\leq j\leq r+\delta-1$.
\end{construction}

\begin{remark}
At first glance there appears to be a distinction between
code symbols $c_{i,j}$ with $1\leq i\leq w-1$ and those with
$i=w$. However, careful thought reveals that the code symbols that
correspond to the sets $S_i$ for $1\leq i\leq w$ are essentially
symmetric, i.e., any $w-1$ sets of code symbols can determine $v$ code
symbols of the remaining set according to \eqref{eqn_global_check}.
\end{remark}

\begin{theorem}\label{thm_optimal_code}
Let $\mu$ be a positive integer, and let $S_i \subseteq \F_q$, $i\in[w]$ be the sets
defined in Construction \ref{cons}. If every subset $\cR\subseteq
\{S_i ~:~ 1\leq i\leq w\}$, $\abs{\cR}=\mu$, satisfies that for all
$S'\in\cR$,
\begin{equation}\label{eqn_cond_mu}
\left|S'\cap \left(\bigcup_{S\in \cR\setminus\{S'\}}S\right)\right|< \delta,
\end{equation}
then the code $\cC$ generated by Construction \ref{cons} is an
$[n,k,d]_q$ linear code, with $d\geq \min\{r-v+\delta,(\mu+1)\delta\}$
and with all symbol $(r,\delta)$-locality, where $n=w(r+\delta-1)$,
$k=(w-1)r+v$, $1\leq v\leq r$, and all parameters are integers.
\end{theorem}
\begin{proof}
By Steps $1$ and $3$, it is easy to check that the code $\cC$
generated by Construction \ref{cons} has all symbol
$(r,\delta)$-locality. By Definition \ref{def_r_delta_i}, the repair sets
are the coordinates of the code symbols $\{f_i(\theta):\theta\in S_i\}$ for
$1\leq i\leq w$. To simplify the notation, instead of define those coordinates, we directly use $S_i$, $1\leq i\leq w$ to denote
the repair sets in this proof.
 The code $\cC$ is an $[n,k]_q$ linear code with
$n=w(r+\delta-1)$ and $k=(w-1)r+v$ according to Construction
\ref{cons}. To complete the proof, we only need to show that $d\geq
d_1=\min\{r-v+\delta, (\mu+1)\delta\} $, i.e., the code $\cC$ can
recover from any $d_1-1$ erasures.

According to the all symbol $(r,\delta)$-locality, it is sufficient to
consider those repair sets containing strictly more
than $\delta-1$ erasures, where for the code $\cC$ the repair sets
correspond to $S_i$ for $1\leq i\leq w$. Since the maximum number
 of erasures we should consider is $d_1-1$, there are at most $\frac{d_1-1}{\delta}$
 repair sets which can have size larger than or equal to $\delta$. Without loss of generality, we
assume that there are $\ell$ sets, $S_1,\dots,S_\ell$, that contain at
least $\delta$ erasures each, and those erasures are located in
coordinates $E_i\subseteq S_i$ for $1\leq i\leq \ell\leq
\frac{d_1-1}{\delta}$. Denote $|E_i|=\tau_i\geq \delta$ for $1\leq
i\leq \ell$ and $\sum_{1\leq i\leq \ell}\tau_i\leq d_1-1\leq
r-v+\delta-1$. In what follows, we prove the claim by induction on
both $\ell$ and the total number of erasures $\sum_{1\leq i\leq
  \ell}\tau_i$.

For the induction base consider the case of $\ell=1$ and $\delta \le |E_1|\leq
d_1-1$. By
Steps $1$ and $3$, we know $f_i(x)$ for $2\leq i\leq w$, i.e.,
$a_{i,t}$ is available for $2\leq i\leq w$ and $1\leq t\leq r-v$. By
\eqref{eqn_global_check}, $a_{1,t}$ for $1\leq t\leq r-v$ can be
calculated. Recall that $|E_1|\leq d_1-1\leq r-v+\delta-1$. We know at
least $v$ values $f_1(\theta)$ for $\theta\in S_1\setminus E_1$, which
together with $f_1(\alpha_t)=a_{1,t}\prod_{1\leq j\leq
  r+\delta-1}(\alpha_t-\theta_{1,j})$ for $1\leq t\leq r-v$ show that
$f_1(x)$ can be recovered. Here we use the fact that $\{\alpha_t ~:~
1\leq t\leq r-v\}\cap S_1=\phi$, i.e., $\prod_{1\leq j\leq
  r+\delta-1}(\alpha_t-\theta_{1,j})\ne 0$.  This is to say, we can
recover all the code symbols $f_1(\theta)$ for $\theta\in E_1$.
We emphasize that in this case, the $S_i$'s are not required to satisfy \eqref{eqn_cond_mu},
so the restriction on the size of the finite field in this case is $q \geq 2r+\delta-v-1$.

For the induction hypothesis assume that for the case $1\leq
\ell=s<\frac{d_1-1}{\delta}$ and $\sum_{1\leq i\leq s}\tau_i=T< d_1-1$,
the code symbols $f_i(\theta)$ for $\theta\in E_i$ and $1\leq i\leq s$
are recoverable.

The induction step is divided into two cases. For the first case,
assume an erasure pattern with $\sum_{1\leq i\leq s}\tau_i=T+1\leq
d_1-1$. Note that if $s=1$ the claim holds by the induction
base. Therefore, we only consider $s\geq 2$. Since $s <
\frac{d_1-1}{\delta}\leq \frac{(\mu+1)\delta-1}{\delta}$, we have
$s\leq \mu.$ Thus, by \eqref{eqn_cond_mu},
$$\left|E_i\cap \left(\bigcup_{\substack{1\leq j\leq s\\ j\ne i}} E_j\right)\right|\leq \left|S_i\cap \left(\bigcup_{\substack{1\leq j\leq s\\ j\ne i}} S_j\right)\right|\leq \delta-1,$$
which means that the elements of each $E_i$ may be indexed $E_i=\{e_{i,t} ~:~ 1\leq t\leq \tau_i\}$ such that
\begin{equation}\label{eqn_reorder}
\{e_{i,t} ~:~ 1\leq t\leq \tau_i-\delta+1\}\cap E_j=\phi\text{ for }1\leq i\ne j\leq s.
\end{equation}
By polynomial interpolation, $f_i(x)$ for $1\leq i\leq s$ with $\deg(f_i(x))\leq r-1$ is represented as
\begin{equation}\label{eqn_f_i_x}
\begin{split}
  f_i(x)&=\sum_{\theta\in S_i\setminus\{e_{i,j} ~:~ \tau_i-\delta+2\leq j\leq \tau_i\}}\frac{f_i(\theta)\prod\limits_{\tau_i-\delta+2\leq j\leq \tau_i}(\theta-e_{i,j})}{\prod\limits_{\theta_1\in S_i\setminus\{\theta\}}(\theta-\theta_1)}\cdot \frac{\prod\limits_{\theta_1\in S_i}(x-\theta_1)}{(x-\theta)\prod\limits_{\tau_i-\delta+2\leq j\leq \tau_i}(x-e_{i,j})}\\
&=\sum_{\theta_{i,t}\in S_i\setminus E_i}\frac{c_{i,t}\prod_{\tau_i-\delta+2\leq j\leq \tau_i}(\theta_{i,t}-e_{i,j})}{\prod_{\theta_1\in S_i\setminus\{\theta_{i,t}\}}(\theta_{i,t}-\theta_1)}\cdot\frac{\prod_{\theta\in S_i}(x-\theta)}{(x-\theta_{i,t})\prod\limits_{\tau_i-\delta+2\leq j\leq \tau_i}(x-e_{i,j})}\\
&\quad+\sum_{1\leq t\leq \tau_i-\delta+1}\varpi_{i,t}\frac{\prod_{\theta\in S_i}(x-\theta)}{(x-e_{i,t})\prod\limits_{\tau_i-\delta+2\leq j\leq \tau_i}(x-e_{i,j})}\\
&=g_i(x)+\sum_{1\leq t\leq \tau_i-\delta+1}\varpi_{i,t}\frac{\prod_{\theta\in S_i}(x-\theta)}{(x-e_{i,t})\prod\limits_{\tau_i-\delta+2\leq j\leq \tau_i}(x-e_{i,j})},
\end{split}
\end{equation}
where $g_i(x)$ is determined by the accessible code symbols
corresponding to $S_i\setminus E_i$ and
$$\varpi_{i,t}=f_i(e_{i,t})\frac{\prod_{\tau_i-\delta+2\leq j\leq \tau_i}(e_{i,t}-e_{i,j})}{\prod_{\theta_1\in S_i\setminus\{e_{i,t}\}}(e_{i,t}-\theta_1)},$$
with ${\prod_{\tau_i-\delta+2\leq j\leq
    \tau_i}(e_{i,t}-e_{i,j})}/{\prod_{\theta_1\in
    S_i\setminus\{e_{i,t}\}}(e_{i,t}-\theta_1)}$ being a nonzero
constant for $1\leq i\leq s$ and $1\leq t\leq
\tau_i-\delta+1$. Combining \eqref{eqn_f_i_x} with
\eqref{eqn_global_check}, we have
\begin{equation}\label{eqn_cond_recov}
\begin{split}
&\left(\varpi_{1,1},\dots,\varpi_{1,\tau_1-\delta+1},\dots,\varpi_{s,\tau_s-\delta+1}\right)M\\
=&\left(\varpi_{1,1},\dots,\varpi_{1,\tau_1-\delta+1},\dots,\varpi_{s,\tau_s-\delta+1}\right)\begin{pmatrix}
 m_{\lambda_{1,1},1} & m_{\lambda_{1,1},2} & \dots  & m_{\lambda_{1,1},r-v} \\
 m_{\lambda_{1,2},1} & m_{\lambda_{1,2},2} & \dots  & m_{\lambda_{1,2},r-v} \\
\vdots & \vdots & \dots  & \vdots \\
 m_{\lambda_{1,\tau_1-\delta+1},1} & m_{\lambda_{1,\tau_1-\delta+1},2} & \dots  & m_{\lambda_{1,\tau_1-\delta+1},r-v} \\
 \vdots & \vdots & \dots  & \vdots \\
 m_{\lambda_{s,\tau_s-\delta+1},1} & m_{\lambda_{s,\tau_s-\delta+1},2} & \dots  & m_{\lambda_{t_{s,\tau_s-\delta+1}},r-v} \\
 \end{pmatrix}_{v_1\times (r-v)}\\
 =&(w_1,w_2,\dots,w_{r-v}),
\end{split}
\end{equation}
where $(w_1,w_2,\dots,w_{r-v})$
 is a constant vector determined by the accessible code symbols with
\begin{equation*}
w_i=-\sum_{1\leq j\leq s}\frac{g_j(\alpha_i)}{\prod_{\theta\in S_j}(\alpha_i-\theta)} -\sum_{s+1\leq j\leq w}\frac{f_j(\alpha_i)}{\prod_{\theta\in S_j}(\alpha_i-\theta)}
\text{ for } 1\leq i\leq r-v,
\end{equation*}
  $v_1=\sum_{1\leq j\leq s}(\tau_i-\delta+1)\leq r-v-(s-1)(\delta-1)<r-v$ and
\begin{equation*}
m_{\lambda_{i,j},z}=\frac{1}{(\alpha_z-e_{i,j})\prod_{\tau_i-\delta+2\leq t\leq \tau_i}(\alpha_z-e_{i,t})}
\end{equation*}
for $1\leq i\leq s$, $1\leq j\leq \tau_i-\delta+1$, and $1\leq z\leq r-v$.

Recall that ${\prod_{\tau_i-\delta+2\leq j\leq
    \tau_i}(e_{i,t}-e_{i,j})}/{\prod_{\theta_1\in
    S_i\setminus\{e_{i,t}\}}(e_{i,t}-\theta_1)}$ is a nonzero constant
for $1\leq i\leq s$ and $1\leq t\leq \tau_i-\delta+1$. Thus,
recovering the vector
$(f_1(e_{1,1}),\dots,f_1(e_{1,\tau_1-\delta+1}),\dots,
f_s(e_{s,\tau_s-\delta+1}))$ is equivalent to recovering the vector
$(\varpi_{1,1},\dots,\varpi_{1,\tau_1-\delta+1},\dots,\varpi_{s,\tau_s-\delta+1})$. Note
that the equation \eqref{eqn_cond_recov} has at least one
solution, namely, the solution that corresponds to the original
codeword.  Thus, by \eqref{eqn_cond_recov},
$(f_1(e_{1,1}),\dots,f_1(e_{1,\tau_1-\delta+1}),\dots,
f_s(e_{s,\tau_s-\delta+1}))$ is recoverable if and only if the
solution is unique, i.e., the rank of $M$ is $v_1$, or equivalently, there exist $v_1$
columns of $M$ forming a non-singular sub-matrix.  Recall that by the
induction hypothesis, the erasure pattern
$E_1,E_2,\dots,E_s\setminus\{e_{s,\tau_s-\delta+1}\}$ is recoverable,
i.e., there exists a $(v_1-1)\times (v_1-1)$ matrix with
\begin{equation}\label{eqn_induc}
  \det
 \begin{pmatrix}
 m_{\lambda_{1,1},t_1} & m_{\lambda_{1,1},t_2} & \dots  & m_{\lambda_{1,1},t_{v_1-1}} \\
 m_{\lambda_{1,2},t_1} & m_{\lambda_{1,2},t_2} & \dots  & m_{\lambda_{1,2},t_{v_1-1}} \\
\vdots & \vdots & \dots  & \vdots \\
 m_{\lambda_{1,\tau_1-\delta+1},t_1} & m_{\lambda_{1,\tau_1-\delta+1},t_2} & \dots  & m_{\lambda_{1,\tau_1-\delta+1},t_{v_1-1}} \\
 \vdots & \vdots & \dots  & \vdots \\
 m_{\lambda_{s,\tau_s-\delta},t_1} & m_{\lambda_{s,\tau_s-\delta},t_2} & \dots  & m_{\lambda_{s,\tau_s-\delta},t_{v_1-1}} \\
 \end{pmatrix}\ne0.
\end{equation}

If the erasure pattern $E_1,E_2,\dots,E_s$ is not recoverable, then each $v_1\times v_1$ sub-matrix of $M$ is singular.
Thus, $\alpha_i$ for $1\leq i\leq r-v$ are roots of $h(x)=0$ with
\begin{equation}\label{eqn_h_x}
h(x)=\det\begin{pmatrix}
 m_{\lambda_{1,1},t_1} & m_{\lambda_{1,1},t_2} & \dots  & m_{\lambda_{1,1},t_{v_1-1}} & m_{\lambda_{1,1}}(x) \\
 m_{\lambda_{1,2},t_1} & m_{\lambda_{1,2},t_2} & \dots  & m_{\lambda_{1,2},t_{v_1-1}} & m_{\lambda_{1,2}}(x) \\
\vdots & \vdots & \dots  & \vdots \\
 m_{\lambda_{1,\tau_1-\delta+1},t_1} & m_{\lambda_{1,\tau_1-\delta+1},t_2} & \dots  & m_{\lambda_{1,\tau_1-\delta+1},t_{v_1-1}} & m_{\lambda_{1,\tau_1-\delta+1}}(x)\\
 \vdots & \vdots & \dots  & \vdots \\
 m_{\lambda_{s,\tau_s-\delta+1},t_1} & m_{\lambda_{s,\tau_s-\delta+1},t_2} & \dots  & m_{\lambda_{s,\tau_s-\delta+1},t_{v_1-1}} & m_{\lambda_{s,\tau_s-\delta+1}}(x)\\
 \end{pmatrix},
\end{equation}
where
\begin{equation}\label{eqn_m_lambda}
m_{\lambda_{i,j}}(x)=\frac{1}{(x-e_{i,j})\prod_{\tau_i-\delta+2\leq t\leq \tau_i}(x-e_{i,t})}\text{ for }1\leq i\leq s\text{ and }1\leq j\leq \tau_i-\delta+1.
\end{equation}
Note that $h(x)\prod_{1\leq u\leq s}\prod_{\theta\in E_u}(x-\theta)$
is a polynomial with degree less than $\sum_{1\leq i\leq
  s}\tau_i-\delta\leq r-v+\delta-1-\delta=r-v-1$ and $\alpha_i$ for
$1\leq i\leq r-v$ are its roots, hence $h(x)\prod_{1\leq u\leq
  s}\prod_{\theta\in E_u}(x-\theta)\equiv0$. However, for $1\leq i,i_1\leq
s$, $1\leq j\leq \tau_i-\delta+1$ and $1\leq j_1\leq
\tau_{i_1}-\delta+1$, \eqref{eqn_reorder} means that
$e_{i,j}\not\in\{e_{i_1,j_1}\}\cup \{e_{i_1,t} ~:~ \tau_i-\delta+2\leq
t\leq \tau_i\}$ when $(i,j)\ne (i_1,j_1)$.  It follows that for $1\leq
i\leq s$ and $1\leq j\leq \tau_i-\delta+1$, $e_{i,j}$ is a root of
$m_{\lambda_{i_1,j_1}}(x)\prod_{1\leq u\leq s}\prod_{\theta\in
  E_u}(x-\theta)=0$ for all $(i_1,j_1)\ne (i,j)$ with $1\leq i_1\leq
s$ and $1\leq j_1\leq \tau_{i_1}-\delta+1$. Again by
\eqref{eqn_reorder}, $e_{i,j}$ for $1\leq i\leq s$ and $1\leq j\leq
\tau_i-\delta+1$ only appears in one of $E_t$ for $1\leq t\leq s$,
i.e.,
$$(x-e_{i,j})\bigg| \prod_{1\leq u\leq s}\prod_{\theta\in
 E_u}(x-\theta),$$
however,
$$(x-e_{i,j})^2\not\bigg|(\prod_{1\leq u\leq s}\prod_{\theta\in
  E_u}(x-\theta)),$$ for $1\leq i\leq s$ and $1\leq j\leq
\tau_i-\delta+1$.  By \eqref{eqn_m_lambda}, we have that $e_{i,j}$ is
not a root of $m_{\lambda_{i,j}}(x)\prod_{1\leq u\leq
  s}\prod_{\theta\in E_u}(x-\theta)=0$ for $1\leq i\leq s$ and $1\leq
j\leq \tau_i-\delta+1$.  Thus, the polynomials
$m_{\lambda_{i,j}}(x)\prod_{1\leq u\leq s}\prod_{\theta\in
  E_u}(x-\theta)$ for $1\leq i\leq s$ and $1\leq j\leq
\tau_i-\delta+1$ are linearly independent over $\F_q$. Therefore,
$h(x)\prod_{1\leq u\leq s}\prod_{\theta\in E_u}(x-\theta)\equiv0$ implies
that the coefficients of $m_{\lambda_{i,j}}(x)\prod_{1\leq u\leq
  s}\prod_{\theta\in E_u}(x-\theta)$ for $1\leq i\leq s$ and $1\leq
j\leq \tau_i-\delta+1$ in $h(x)\prod_{1\leq u\leq s}\prod_{\theta\in
  E_u}(x-\theta)$ are $0$. This is to say, the coefficient of
$m_{\lambda_{s,\tau_{s}-\delta+1}}(x)\prod_{1\leq u\leq
  s}\prod_{\theta\in E_u}(x-\theta)$ in $h(x)\prod_{1\leq u\leq
  s}\prod_{\theta\in E_u}(x-\theta)$ is zero, i.e.,
\begin{equation*}
\det\begin{pmatrix}
 m_{\lambda_{1,1},t_1} & m_{\lambda_{1,1},t_2} & \dots  & m_{\lambda_{1,1},t_{v_1-1}} \\
 m_{\lambda_{1,2},t_1} & m_{\lambda_{1,2},t_2} & \dots  & m_{\lambda_{1,2},t_{v_1-1}} \\
\vdots & \vdots & \dots  & \vdots \\
 m_{\lambda_{1,\tau_1-\delta+1},t_1} & m_{\lambda_{1,\tau_1-\delta+1},t_2} & \dots  & m_{\lambda_{1,\tau_1-\delta+1},t_{v_1-1}} \\
 \vdots & \vdots & \dots  & \vdots \\
 m_{\lambda_{s,\tau_s-\delta},t_1} & m_{\lambda_{s,\tau_s-\delta},t_2} & \dots  & m_{\lambda_{s,\tau_s-\delta},t_{v_1-1}} \\
 \end{pmatrix}=0,
 \end{equation*}
 which is a contradiction with \eqref{eqn_induc}. Thus, the erasure
 pattern $E_1,E_2,\dots,E_s$ is also recoverable.

For the second case of the induction step, assume $\ell=s+1\leq
\frac{d_1-1}{\delta}$ sets and $|E_{s+1}|=\delta$, when
$T<d_1-\delta\leq r-v$.  In this case, by a similar analysis, we have $s+1\leq
\mu$, and thus we also have
\begin{equation*}
\{e_{i,t} ~:~ 1\leq t\leq \tau_i-\delta+1\}\cap E_j=\phi\text{ for }1\leq i\ne j\leq s+1,
\end{equation*}
with $E_i=\{e_{i,t} ~:~ 1\leq t\leq \tau_i\}$ for $1\leq i\leq s+1$,
and
\begin{equation*}
\begin{split}
&\left(\varpi_{1,1},\dots,\varpi_{1,\tau_1-\delta+1},\dots,\varpi_{s,\tau_s-\delta+1},\varpi_{s+1,1}\right)M_{s+1}\\
=&\left(\varpi_{1,1},\dots,\varpi_{1,\tau_1-\delta+1},\dots,\varpi_{s,\tau_s-\delta+1},\varpi_{s+1,1}\right)
\begin{pmatrix}
 m_{\lambda_{1,1},1} & m_{\lambda_{1,1},2} & \dots  & m_{\lambda_{1,1},r-v} \\
 m_{\lambda_{1,2},1} & m_{\lambda_{1,2},2} & \dots  & m_{\lambda_{1,2},r-v} \\
\vdots & \vdots & \dots  & \vdots \\
 m_{\lambda_{1,\tau_1-\delta+1},1} & m_{\lambda_{1,\tau_1-\delta+1},2} & \dots  & m_{\lambda_{1,\tau_1-\delta+1},r-v} \\
 \vdots & \vdots & \dots  & \vdots \\
 m_{\lambda_{s+1,1},1} & m_{\lambda_{s+1,1},2} & \dots  & m_{\lambda_{t_{s+1,1}},r-v} \\
 \end{pmatrix}_{v_2\times (r-v)}\\
 =&(w_1,w_2,\dots,w_{r-v}),
\end{split}
\end{equation*}
where $(w_1,w_2,\dots,w_{r-v})$ is a constant vector determined by the accessible code symbols, $v_2=\sum_{1\leq j\leq s+1}(\tau_i-\delta+1)\leq T+\delta-(s+1)(\delta-1)
<r-v+1-s(\delta-1)\leq r-v$, and
\begin{equation*}
m_{\lambda_{i,j},z}=\frac{1}{(\alpha_z-e_{i,j})\prod_{\tau_i-\delta+2\leq t\leq \tau_i}(\alpha_z-e_{i,t})}
\end{equation*}
for $1\leq i\leq s+1$, $1\leq j\leq \tau_i+\delta-1$, and $1\leq z\leq
r-v$.  Again by the induction hypothesis, there should exists a
$(v_2-1)\times (v_2-1)$ matrix with
\begin{equation}\label{eqn_induc_2}
 \det\begin{pmatrix}
 m_{\lambda_{1,1},t_1} & m_{\lambda_{1,1},t_2} & \dots  & m_{\lambda_{1,1},t_{v_2-1}} \\
 m_{\lambda_{1,2},t_1} & m_{\lambda_{1,2},t_2} & \dots  & m_{\lambda_{1,2},t_{v_2-1}} \\
\vdots & \vdots & \dots  & \vdots \\
 m_{\lambda_{1,\tau_1-\delta+1},t_1} & m_{\lambda_{1,\tau_1-\delta+1},t_2} & \dots  & m_{\lambda_{1,\tau_1-\delta+1},t_{v_2-1}} \\
 \vdots & \vdots & \dots  & \vdots \\
 m_{\lambda_{s,\tau_s-\delta},t_1} & m_{\lambda_{s,\tau_s-\delta},t_2} & \dots  & m_{\lambda_{s,\tau_s-\delta},t_{v_2-1}} \\
 \end{pmatrix}\ne0,
\end{equation}
i.e., the erasure pattern
$E_1,E_2,\dots,E_s,(E_{s+1}\setminus\{e_{s+1,1}\})$ is
recoverable. Here, $f_{s+1}(\theta)$ for $\theta\in
E_{s+1}\setminus\{e_{s+1,1}\}$ is recovered by the
$(r,\delta)$-locality independently, since
$|E_{s+1}\setminus\{e_{s+1,1}\}|=\delta-1$. If $E_1,E_2,\dots,E_{s+1}$
is not recoverable, then all the $v_2\times v_2$ sub-matrices of $M_{s+1}$
are singular.  Therefore, by the same analysis, the polynomials
$m_{\lambda_{i,j}}(x)\prod_{1\leq u\leq s+1}\prod_{\theta\in
  E_u}(x-\theta)$ for $1\leq i\leq s+1$ and $1\leq j\leq
\tau_i-\delta+1$ are linearly independent over $\F_q$. This is also a
contradiction with \eqref{eqn_induc_2} and all the $v_2\times v_2$
sub-matrices of $M_{s+1}$ are singular, by the same analysis as the
previous case. Thus, the erasure pattern $E_1,E_2,\dots,E_{s+1}$ is
also recoverable.

Therefore, by mathematical induction, the distance of $\cC$ satisfies
$d\geq d_1$, which completes the proof.
\end{proof}

\subsection{Explicit locally repairable codes with $n>q$}

According to the bound of Lemma \ref{lemma_bound_i}, the minimal Hamming
distance of the code $\cC$ generated by Construction \ref{cons}, i.e, $n=w(r+\delta-1)$ and $k=(w-1)r+v$ for $0< v\leq r$, is at
most $r-v+\delta$. In fact,
the key point in applying Theorem \ref{thm_optimal_code} is to find sets
$S_1,\dots,S_w$ of evaluation points, that both allow optimal code
construction with  the minimal Hamming distance $d=r-v+\delta$ as well a long code. In this subsection, based on Construction A,
we analyze special structures of $S_1,\dots,S_w$ that can yield optimal locally
repairable codes with $n>q$.

\vspace{2mm}
\noindent\textbf{Two trivial optimal locally repairable codes with $n>q$}

\begin{corollary}\label{corollary_no_limited_length}
  Let $n=w(r+\delta-1)$, $k=(w-1)r+v$, $1\leq v\leq r$, be
  integers. If $r-v \leq \delta$ and $q\geq 2r+\delta-v-1$, then there exists
  an optimal $[n,k,d=r-v+\delta]_q$ linear code with all symbol
  $(r,\delta)$-locality, where optimality is with respect to the bound
  in Lemma \ref{lemma_bound_i}.
\end{corollary}
\begin{proof}
  By Lemma \ref{lemma_bound_i}, a code with the given $n$ and $k$ is
  optimal if $d=r-v+\delta$.  Since $r-v \leq \delta$, in the proof of
  Theorem \ref{thm_optimal_code} we only need to consider the case
  that there is only one repair set containing strictly more than $\delta-1$ erasures, which easily holds.
\end{proof}

\begin{remark}
  We remark that in the case described in Corollary
  \ref{corollary_no_limited_length}, we can let $S_i=S_j$ for $1\leq
  i\ne j\leq w$.  In this case,  $r-v \le \delta$
  and $q \ge 2r+\delta-v-1$ are sufficient for the code
  generated by Construction A to be optimal. This is to say, the value
  $w$ is independent of $q$.
  Thus, the length $n=w(r+\delta-1)$ of the code $\cC$ can be as long as
  we wish. This result is already known for the case $\delta=2$ and
  $d\leq 4$ (see \cite{LXY}), and is, to the best of our knowledge,
  new for the case of $\delta>2$. This result also shows that the condition
  $2t+1>4$ is necessary for Theorem \ref{theorem_bound_delta>2}, since the code
   length is unbounded for the
  case $2t+1\leq 4$, i.e., $t\leq 1$ corresponding to the case $r-v\leq\delta$, where $t=\lfloor(d-1)/\delta\rfloor=\lfloor\frac{r+v+\delta-1}{\delta}\rfloor$.
\end{remark}

\begin{corollary}\label{corollary_optimal_flexible}
  Let $n=w(r+\delta-1)$, $k=(w-1)r+v$, $1\leq v\leq r$, be
  integers. Let $S\subseteq\F_q\setminus\{\alpha_i ~:~ 1\leq i\leq
  r-v\}$, $\abs{S}=\delta-1$, be a fixed subset. Take $S_i\subseteq
  \F_q\setminus\{\alpha_i ~:~ 1\leq i\leq r-v\}$ for $1\leq i\leq
  w$. If $S_i\cap S_j\subseteq S$ for $1\leq i\ne j\leq w$, then the
  code $\cC$ generated by Construction \ref{cons} is an optimal
  $[n,k,d=r-v+\delta]_q$ linear code with all symbol
  $(r,\delta)$-locality, where optimality is with respect to the bound
  in Lemma \ref{lemma_bound_i}.
\end{corollary}

\begin{corollary}\label{corollary_optimal_flexible-m1}
  Let $n=w(r+\delta-1)$, $k=(w-1)r+v$, $1\leq v\leq r$,  be
  integers. If $q \ge (w+1)r+\delta-v-1$, then there exists an optimal
  $[n,k,d=r-v+\delta]_q$ linear code with all symbol
  $(r,\delta)$-locality, where optimality is with respect to the bound
  in Lemma \ref{lemma_bound_i}.
\end{corollary}
\begin{proof}
  When $q \ge (w+1)r+\delta-v-1$, those $S_i$'s in Corollary \ref{corollary_optimal_flexible}
  can be easily constructed by letting $|S|=\delta-1$ and $S_i \cap S_j = S$ for all $1 \le i \ne j \le w$,
  which form a sunflower with center $S$.
\end{proof}

\begin{remark}
 When $w > 1 + \frac{r-v}{\delta-1}$,
  the optimal linear codes with all symbol
  $(r,\delta)$-locality in Corollary \ref{corollary_optimal_flexible-m1} are all with $n>q$.
  In \cite{KBTY}, optimal locally repairable codes are also
  constructed with flexible parameters. However, in \cite{KBTY} the
  construction is based on the so-called good polynomials \cite{TB,LMC} and $n\leq q$.
\end{remark}

\vspace{2mm}
\noindent\textbf{Optimal locally repairable codes with $n>q$ based on union-intersection-bounded family}

A combinatorial structure that captures the interaction between the
evaluation-point sets, $S_1,\dots,S_w$, in Construction \ref{cons} is
a union-intersection-bounded family \cite{GM}. Its definition is now
given:

\begin{definition}[\cite{GM}]\label{def_UIBF}
Let $n_1,\tau,\delta,t,s$ be positive integers such that $n_1\geq
\tau\geq 2$, $\tau\geq \delta$ and $t\geq s$. The
$(s,t;\delta)$-union-intersection-bounded family (denoted by
$(s,t;\delta)$-UIBF$(\tau,n_1)$) is a pair $(\cX,\cS)$, where $\cX$ is a
set of $n_1$ elements (called points) and $\cS\subseteq 2^\cX$ is a
collection of $\tau$-subsets of $\cX$ (called blocks), such that any
$s+t$ distinct blocks $A_1,A_2,\dots,A_s,B_1,B_2,\dots,B_t\in \cS$
satisfy
$$\left|\left(\bigcup_{1\leq i\leq s}A_i\right)\bigcap \left(\bigcup_{1\leq i\leq t}B_i\right)\right|<\delta.$$
\end{definition}

The following corollary follows from Theorem
\ref{thm_optimal_code} and Lemma \ref{lemma_bound_i}.

\begin{corollary}\label{corollary_optimal_code}
Let $n=w(r+\delta-1)$, $k=(w-1)r+v$, $1\leq v\leq r$, be integers, and
let $\mu$ be a positive integer with $\mu\delta\geq r-v$.  If
$(\F_q\setminus\{\alpha_t ~:~ 1\leq t\leq r-v\},\cS=\{S_i ~:~ 1\leq
i\leq w\})$ is a $(1,\mu-1;\delta)$-UIBF$(r+\delta-1,q-r+v)$ then the
code $\cC$ generated by Construction \ref{cons} is an optimal
$[n,k,d=r-v+\delta]_q$ linear code with all symbol
$(r,\delta)$-locality, where optimality is with respect to the bound in Lemma \ref{lemma_bound_i}.
\end{corollary}
\begin{proof}
By Definition \ref{def_UIBF}, each $\mu$-subset $\cR\subseteq\cS$ satisfies that for any $S'\in\cR$,
\begin{equation*}
\left|S'\bigcap \left(\bigcup_{S\in \cR\setminus\{S'\}}S\right)\right|< \delta.
\end{equation*}
By Lemma \ref{lemma_bound_i} we have $d\leq r-v+\delta$.  Thus, the desired
conclusion follows from Theorem \ref{thm_optimal_code} and Lemma
\ref{lemma_bound_i}.
\end{proof}

In \cite{GM}, a lower bound on the size of
$(1,\mu-1;\delta)$-UIBF$(r+\delta-1,q)$ is given, which immediately
implies a lower bound on the length of the codes generated by
Construction \ref{cons} according to Corollary
\ref{corollary_optimal_code}.

\begin{lemma}[\cite{GM}]\label{lemma_lower_B_UIBF}
Let $\mu,\delta,r, n_1$ be positive integers. Then there exists a
$(1,\mu-1;\delta)$-UIBF$(r+\delta-1,n_1)$ $(\cX,\cS)$ with
$|\cS|=\Omega({n_1}^{\frac{\delta}{\mu-1}})$, where $r$, $\delta$, $\mu$ are
regarded as constants.
\end{lemma}

Based on Corollary \ref{corollary_optimal_code} and Lemma \ref{lemma_lower_B_UIBF}, we have the following:

\begin{corollary}\label{corollary_lower_B_length}
Let $n=w(r+\delta-1)$, $k=(w-1)r+v$, $1\leq v\leq r$, be integers, and
let $\mu$ be a positive integer with $\mu\delta\geq r-v$. Then
Construction \ref{cons} can generate an optimal (with respect to the
bound in Lemma \ref{lemma_bound_i}) $[n,k,d=r-v+\delta]_q$ linear code
$\cC$ with all symbol $(r,\delta)$-locality and length
$n=\Omega(q^{\frac{\delta}{\mu-1}})$ where we regard $r$, $\delta$, and $\mu$ as constants.
\end{corollary}

\vspace{2mm}
\noindent\textbf{Optimal locally repairable codes with $n>q$ based on packings or Steiner systems}

In the following, we consider some special sufficient conditions for \eqref{eqn_cond_mu}
to construct optimal linear codes with all symbol
$(r,\delta)$-locality.

\begin{theorem}\label{theorem_optimal_code_cond_a}
Let $n=w(r+\delta-1)$, $k=(w-1)r+v$, $1\leq v\leq r$, be integers, and
let $a$ be a positive integer.  If $|S_i\cap S_j|\leq a$ for $1\leq i\ne j\leq w$ and
$r-v\le\frac{\delta^2}{a}$, then the code $\cC$
generated by Construction \ref{cons} is an optimal
$[n,k,d=r-v+\delta]_q$ linear code with all symbol
$(r,\delta)$-locality, where optimality is with respect to the bound
in Lemma \ref{lemma_bound_i}.
\end{theorem}
\begin{proof}
Denote $\cS=\mathset{S_1,\dots,S_w}$, and let
$\mu=\lceil\frac{\delta}{a}\rceil$. Then the fact that $|S_i\cap
S_j|\leq a$ means that for any $\mu$-subset, $\cR\subseteq\cS$, and
for any $S'\in\cR$, we have
\begin{equation*}
  \left|S'\cap \left(\bigcup_{S\in\cR\setminus\mathset{S'}}S\right)\right|\leq
  (\mu-1)a=\left(\left\lceil\frac{\delta}{a}\right\rceil-1\right)a
  \leq \delta-1.
\end{equation*}
Since $\mu\delta\geq \frac{\delta^2}{a}\ge r-v$, the conclusion follows by Theorem \ref{thm_optimal_code}.
\end{proof}

\begin{definition}(\cite[VI. 40]{CD})\label{def_Packing}
Let $n_1\geq 2$ be an integer and $u$ a positive integer.  A
$\tau$-$(n_1,t,1)$-\textit{packing} is a pair $(\cX,\cS)$, where $\cX$ is
a set of $n_1$ elements (called points) and $\cS\subseteq 2^\cX$ is a
collection of $t$-subsets of $\cX$ (called blocks), such that each
$\tau$-subset of $\cX$ is contained in at most one block of
$\cS$. Furthermore, if each $\tau$-subset of $\cX$ is contained in
exactly one block of $\cS$, then $(\cX,\cS)$ is also called a
$(\tau,t,n_1)$-\textit{Steiner system}.
\end{definition}

The following corollary follows directly from Theorem
\ref{theorem_optimal_code_cond_a}.
\begin{corollary}\label{corollary_optimal_code_packing}
Let $n_1=q-r+v$. If there exists a
$(\tau+1)$-$(n_1,r+\delta-1,1)$-packing with blocks $\cS$
and $0 \le r-v \le \frac{\delta^2}{\tau}$, then
there exists an optimal $[n,k,d]_q$ linear
code with all symbol $(r,\delta)$-locality, where
$n=|\cS|(r+\delta-1)$, $k=(|\cS|-1)r+v$, and $d=r-v+\delta$.
\end{corollary}

The number of blocks of a packing is upper bounded by the following Johnson bound \cite{J}:

\begin{lemma}[\cite{J}]\label{lemma_johnson}
  The maximum possible number of blocks of a
  $(\tau+1)$-$(n_1,r+\delta-1,1)$-packing $\cS$ is bounded by
\begin{equation*}
\begin{split}
|\cS|
\leq\left\lfloor \frac{n_1}{r+\delta-1} \left\lfloor\frac{n_1-1}{r+\delta-2} \left\lfloor\frac{n_1-2}{r+\delta-3}
\dots\left\lfloor\frac{n_1-\tau}{r+\delta-1-\lambda} \right\rfloor\dots\right\rfloor\right\rfloor\right\rfloor.
\end{split}
\end{equation*}
\end{lemma}

Thus, the number of blocks for a
$(\tau+1)$-$(n_1,r+\delta-1,1)$-packing can be as large as
$O(n_1^{\tau+1})$, when $\tau$, $r$, and $\delta$ are regarded as constants.

\begin{corollary}\label{corollary_code_via_packings}
Let $n_1=q-r+v$. If there exists a
$(\tau+1)$-$(n_1,r+\delta-1,1)$-packing with blocks $\cS$,
$|\cS|=O(n_1^{\tau+1})$,
and $0 \le r-v\le \frac{\delta^2}{\tau}$,
then there
exists an optimal $[n,k,d]_q$ linear code with all symbol
$(r,\delta)$-locality, where $n=|\cS|(r+\delta-1)=O(q^{\tau+1})$,
$k=(|\cS|-1)r+v$ and $d=r-v+\delta$. In particular, for the case
$w-1\geq 2(r-v+1)$, $r-v=\delta+1$, i.e., $d=2\delta+1$ and
$\tau=\delta-1$, the code based on the
$(\tau+1)$-$(n_1,r+\delta-1,1)$-packing has asymptotically optimal
length, where $r$ and $\delta$ are regarded as constants.
\end{corollary}
\begin{proof}
By Corollary \ref{corollary_optimal_code_packing}, we have
$n=|\cS|(r+\delta-1)=O(q^{\tau+1})$ for the code generated by
Construction \ref{cons}. For the case $r-v=\delta+1$, $w-1\geq
2(r-v+1)$, $d=2\delta+1$, and $t=\floorenv{(d-1)/\delta}=2$, by Theorem \ref{theorem_bound_delta>2} we have
$$n\leq
\frac{t(r+\delta-1)}{2r(q-1)}q^{\frac{2(w-w+1)r-2v}{t}}\leq\frac{t(r+\delta-1)}{2r(q-1)}q^{r-v}=O(q^{r-v-1}).$$
Thus, for the case $r-v=\delta+1$ and $\tau=\delta-1$, the code $\cC$
has length $n=O(q^{\tau+1})=O(q^{\delta})$, which is asymptotically
optimal with respect to the bound in Theorem
\ref{theorem_bound_delta>2}, when $r$ and $\delta$ are regarded as
constants.
\end{proof}

As an example, we also analyze the length of the codes based on
Steiner systems.

\begin{corollary}
Let $n_1=q-r+v$. If there exists a $(\tau+1,r+\delta-1,n_1)$-Steiner
system and $0\le r-v\leq \frac{\delta^2}{\tau}$,
then there exists an
optimal $[n,k,d]_q$ linear code with all symbol $(r,\delta)$-locality,
where
\begin{align*}
n&=\frac{\binom{n_1}{\tau+1}(r+\delta-1)}{\binom{r+\delta-1}{\tau+1}},\\
k&=\parenv{\frac{\binom{n_1}{\tau+1}}{\binom{r+\delta-1}{\tau+1}}-1}r+v,
\end{align*}
and $d=r-v+\delta$.  In particular, for the case $w-1\geq 2(r-v+1)$,
$r-v=\delta+1$, i.e., $d=2\delta+1$ and $\tau=\delta-1$, the code
based on the $(\delta,r+\delta-1,q-\delta-1)$-Steiner system has asymptotically
optimal length, where $r$ and $\delta$ are regarded as constants.
\end{corollary}
\begin{proof}
The first part of the corollary follows directly from Corollary
\ref{corollary_optimal_code_packing} and Definition
\ref{def_Packing}. For the second part, the fact $\tau=\delta-1$ means
that $r-v=\delta+1<\frac{\delta^2}{\delta-1}$ is possible, which also
means the code $\cC$ has length
$(r+\delta-1)\binom{q-\delta+1}{\delta}/\binom{r+\delta-1}{\delta}$
and $d=2\delta+1$. Since $w-1\geq 2(r-v+1)$, $u=w-1$, $r-v=\delta-1$, and $d=2\delta+1$, i.e., $t=2$, by Theorem \ref{theorem_bound_delta>2}, we have
$$n\leq
\frac{t(r+\delta-1)}{2r(q-1)}q^{\frac{2(w-u)r-2v}{t}}\leq\frac{t(r+\delta-1)}{2r(q-1)}q^{r-v}=O(q^{\delta}).$$
Now the conclusion comes from the fact
that the upper bound is $O(q^\delta)$ and the constructed code has length
$n=\Omega(q^\delta)$,
where we assume $r$ and $\delta$ are constants.
\end{proof}

\begin{remark}
For the case $\delta=2$ and $d=5$, optimal linear codes with all
symbol $(r,2)$-locality and asymptotically optimal length
$\Theta(q^2)$ have been introduced in \cite{GXY,Jin,BCGLP}.
\end{remark}
\begin{remark}
  Given positive integers $\tau$, $r$ and $\delta>2$, the natural
    necessary conditions for the existence of a
$(\tau+1,r+\delta-1,q-r+v)$-Steiner system are that
$\binom{q-r+v-i}{\tau+1-i}|\binom{r+\delta-1-i}{\tau+1-i}$ for
all $0\leq i\leq \tau$. It was shown in \cite{K} that these
  conditions are also sufficient except perhaps for finitely many
  cases. While $q$ might not be a prime power, any prime power $\overline{q}\geq
  q$ will suffice for our needs. It is known, for example, that there
  is always a prime in the interval $[q,q+q^{21/40}]$ (see
  \cite{BHP}). Thus, Construction \ref{cons} provides infinitely many
  optimal linear $[n,k,d]_{\overline{q}}$ locally repairable codes, with all symbol
  $(r,\delta)$-locality, and
  \begin{align*}
   n&=(r+\delta-1)\cdot\frac{\binom{q-r+v}{\tau+1}}{\binom{r+\delta-1}{\tau+1}}=\Omega(q^{\tau+1})=\Omega(\overline{q}^{\tau+1}),\\
    k&= \parenv{\frac{\binom{q-r+v}{\tau+1}}{\binom{r+\delta-1}{\tau+1}}-1}r+v,\\
    d&=r-v+\delta,
  \end{align*}
  i.e., with length super-linear in the field size.
\end{remark}

\section{Concluding Remarks}\label{sec-conclusion}
In this paper, we first derived an upper bound for the length of
optimal locally repairable codes when $\delta>2$. As a byproduct, we
also extended the range of parameters for the known bound (the case
$\delta=2$) and improve its performance for the case $d>r+\delta$.
A general construction of locally repairable codes was
introduced. By the construction, locally repairable codes with length
super-linear in the field size can be generated. In particular, for
some cases those codes have asymptotically optimal length with
respect to the new bound.

Several combinatorial structures, e.g., union-intersection-bounded
families,
packings, and Steiner systems, satisfy \eqref{eqn_cond_mu}
and play a key role in determining the length of the codes generated
by Construction \ref{cons}.  If more of those structures with a large
number of blocks can be constructed, more good codes with length $n>q$
can be generated. Finding more such combinatorial structures and
explicit constructions for them, is left for future research.

\section*{Appendix}

\noindent\textbf{Proof of Lemma \ref{lemma_for_D(B)}}

 We first construct a uniform $\ocB$ from $\cB$, by arbitrarily
  adding elements to sets in $\cB$ that contain less than $r+\delta-1$
  elements. Note that $\ocB$ is not necessarily an ECF. Obviously
  $D(\ocB)\geq D(\cB)$. We contend now that $D(\ocB)>0$. If
  $D(\cB)\neq 0$ this is immediate, since we have $D(\ocB)\geq D(\cB)
  > 0$. If $\cB$ is not uniform, at least one set $B\in\cB$ has
  $\abs{B}<r+\delta-1$, and adding elements to it in the process of
  creating $\ocB$ necessarily increases the overlap, i.e.,
  $D(\ocB)>D(\cB)\geq 0$. We also observe that,
  \begin{eqnarray*}
    D(\ocB)=\sum_{\oB\in \ocB }|\oB|-\left|\bigcup_{\oB\in \ocB }\oB\right|=
    |\ocB|(r+\delta-1)-n\equiv -m\pmod{r+\delta-1}.
  \end{eqnarray*}

Next, we partition $\overline{\mathcal{B}}$ into two subsets, $\overline{\mathcal{B}}_1$ and $\overline{\mathcal{B}}_2$, where
\begin{equation*}\label{eqn_def_B_1}
\overline{\mathcal{B}}_1=\{\overline{B}\in\overline{\mathcal{B}} ~:~\exists \overline{B}'\in\overline{\mathcal{B}}, \overline{B}'\neq \overline{B}, \overline{B}\cap \overline{B}'\ne \emptyset\},
\end{equation*}
and
\begin{equation*}
\overline{\mathcal{B}}_2=\overline{\mathcal{B}}\setminus \overline{\mathcal{B}}_1.
\end{equation*}
For convenience, denote $\overline{\mathcal{B}}_1=\{\overline{B}_{1},\dots,\overline{B}_{K}\}$ and
$\overline{\mathcal{B}}_2=\{\overline{B}_{{K+1}},\dots,\overline{B}_{T}\}$ where $0\le K\le T$.

Let $1\le t\le T$ be a positive integer.
Obviously, if $t\ge K$, then $\overline{\mathcal{B}}'=\{\overline{B}_{1},\dots,\overline{B}_{K},\dots,
\overline{B}_{t}\}$ is a $t$-subset satisfying
\begin{equation}\label{Eqn_DB1}
D(\overline{\mathcal{B}}')=\sum_{i=1}^{t}|\overline{B}_{i}|-\left|\bigcup_{i=1}^{t}\overline{B}_{i}\right|=  D(\overline{\mathcal{B}}).
\end{equation}
For the case $0\leq t\leq 1$, the fact $\lfloor t/2\rfloor=0$ means
that the lemma follows trivially.
For the case $2\leq t < K$, we claim that we can select a $t$-subset
$\overline{\mathcal{B}}'\subseteq\overline{\mathcal{B}}_1$ containing $\lfloor t/2\rfloor$ different pairs
of sets $\{\overline{B}_{\tau_{2i-1}},\overline{B}_{\tau_{2i}}\}$ for $1\leq i\leq \lfloor
t/2\rfloor$ with
\begin{eqnarray*}
\sum_{\overline{B}\in {\cB}_j}|\overline{B}|-\left|\bigcup_{\overline{B}\in {\cB}_j}\overline{B}\right|
&\ge&1+\sum_{\overline{B}\in {\cB}_{j-1}}|\overline{B}|-\left|\bigcup_{\overline{B}\in {\cB}_{j-1}}\overline{B}\right|\\
&\ge& j,
\end{eqnarray*}
for ${\cB}_0=\emptyset$ and ${\cB}_j=\{\overline{B}_{\tau_i}: 1\le i\le 2j\}$,  $1\leq j\leq \left\lfloor\frac{t}{2}\right\rfloor$, especially
$\overline{\mathcal{B}}'\supseteq{\cB}_{\left\lfloor\frac{t}{2}\right\rfloor}$ satisfying
\begin{eqnarray}\label{Eqn_DB2}
\sum_{\overline{B}\in \overline{\mathcal{B}}'}|\overline{B}|-\left|\bigcup_{\overline{B}\in\overline{\mathcal{B}}'}\overline{B}\right|
&\ge&\sum_{\overline{B}\in {\cB}_{\left\lfloor\frac{t}{2}\right\rfloor}}|\overline{B}|-\left|\bigcup_{\overline{B}\in {\cB}_{\left\lfloor\frac{t}{2}\right\rfloor}}\overline{B}\right|\nonumber\\
&\ge& \left\lfloor\frac{t}{2}\right\rfloor.
\end{eqnarray}

Otherwise, there exists a subset
$\overline{\mathcal{B}}_1^*\subseteq\overline{\mathcal{B}}_1$ with size at most
$2(\lfloor\frac{t}{2}\rfloor-1)$ such that for any $\overline{B}'\in
\overline{\mathcal{B}}_1\setminus \overline{\mathcal{B}}^*_1, \overline{B}''\in \overline{\mathcal{B}}_1$,
$$\sum_{\overline{B}\in \overline{\mathcal{B}}^*_{1}\cup\{\overline{B}',\overline{B}''\}}|B|-\abs{\bigcup_{\overline{B}\in \overline{\mathcal{B}}^*_{1}\cup\{\overline{B}',\overline{B}''\}}\overline{B}}\leq\sum_{\overline{B}\in \overline{\mathcal{B}}^*_{1}}|\overline{B}|-\abs{\bigcup_{\overline{B}\in \overline{\mathcal{B}}^*_1}\overline{B}},$$
which implies
\begin{eqnarray*}
\left\{\begin{array}{ll}
|\overline{B}'|+|\overline{B}''|\leq \left|(\overline{B}'\cup \overline{B}'')\setminus\bigcup_{\overline{B}\in\overline{\mathcal{B}}^*_1} \overline{B}\right|, & \mathrm{if}~\overline{B}''\in\overline{\mathcal{B}}_1\setminus \overline{\mathcal{B}}^*_1,\\
|\overline{B}'|\leq \left|\overline{B}'\setminus\bigcup_{\overline{B}\in\overline{\mathcal{B}}^*_1} \overline{B}\right|, & \mathrm{if}~\overline{B}''\in\overline{\mathcal{B}}^*_1.\\
\end{array}
\right.
\end{eqnarray*}
However, this means that every $\overline{B}'\in\overline{\mathcal{B}}_1\setminus \overline{\mathcal{B}}^*_1$ has an
empty intersection with any other set in $\overline{\mathcal{B}}_1$, which contradicts
the definition of $\overline{\mathcal{B}}_1$.

By combining \eqref{Eqn_DB1} and \eqref{Eqn_DB2}, for any given $0\leq t\leq |\cB|$, there
  exists a $t$-subset, say
  $\ocB'=\mathset{\oB_1,\oB_2,\dots,\oB_{t}}\subseteq \ocB$, such that
  \begin{equation}\label{eqn_u_sets_*}
    D(\ocB')=\sum_{\oB\in \ocB'}|\oB|-\left|\bigcup_{\oB\in \ocB'}\oB\right|\geq \min \left\{D(\ocB),\left\lfloor{t}/{2}\right\rfloor\right\}\geq \min\left\{r+\delta-1-m,\left\lfloor t/2\right\rfloor\right\},
  \end{equation}
  where the last inequality holds since $D(\ocB)> 0$ and
  $D(\ocB)\equiv -m \pmod{r+\delta-1}$.

  If $\oB_i\in\ocB'$ was created from $B_i\in\cB$, i.e.,
  $B_i\subseteq\oB_i$, then by \eqref{eqn_u_sets_*} we have,
\begin{eqnarray*}
t(r+\delta-1)-\left|\bigcup_{i=1}^{t}B_i\right|
=\sum_{i=1}^{t}|\oB_i|-\left|\bigcup_{i=1}^{t}B_i\right|
\geq \sum_{i=1}^{t}|\oB_i|-\left|\bigcup_{i=1}^{t}\oB_i\right|
\geq \min\left\{r+\delta-1-m,\left\lfloor t/2\right\rfloor\right\}.
\end{eqnarray*}
Now set $\cB'=\mathset{B_1,\dots,B_t}$ to complete the proof. \qed

\vspace{2mm}
\noindent\textbf{Proof of Lemma \ref{lemma_initialization}}

 By Definition \ref{def_r_delta_i}, $\Gamma$ contains at least one repair set for each
  code symbol, hence
  \begin{equation}\label{eqn_Gamma} \bigcup_{R\in\Gamma}R=[n].
  \end{equation}
  If for each $R\in \Gamma$, $R\not\subseteq \bigcup_{R'\in
    \Gamma\setminus\{R\}}R'$, then set $\cR=\Gamma$ and the lemma
  follows.  Otherwise, set $\Gamma_1=\Gamma\setminus\{R\}$, where $R \in \Gamma$ satisfies that
  $R\subseteq \bigcup_{R'\in\Gamma\setminus\{R\}}R'$. Thus, by \eqref{eqn_Gamma}, we
  conclude that
\begin{equation*}
\bigcup_{R'\in \Gamma\setminus\{R\}}R'=[n].
\end{equation*}
 Since $|\Gamma_1|<|\Gamma|$, and
$\Gamma_1$ also satisfies \eqref{eqn_Gamma}, we can repeat the
elimination procedure to obtain the desired set $\cR$. The facts
$\rank(\bigcup_{R\in \cR}R)=k$ and $\rank(R)\leq r$ imply that $|\cR|\geq \lceil\frac{k}{r}\rceil$,
which completes the proof.\qed

\vspace{2mm}
\noindent\textbf{Proof of Lemma \ref{lemma_find_V_1}}

Before proving  Lemma \ref{lemma_find_V_1}, we need to discuss the structures of the repair sets
in more details in three lemmas.

\begin{lemma}\label{lemma_for_rank_B_i}
  Let $\cC$ be an $[n,k]_q$ linear code with all symbol
  $(r,\delta)$-locality. Let $\cR$ be the ECF given by Lemma
  \ref{lemma_initialization}. If for a subset $\cV\subseteq \cR$, and for
  all $R'\in\cV$,
 \begin{equation}\label{eqn_condition_joint}
 \left|R'\bigcap \parenv{\bigcup_{R\in \cV\setminus\{R'\}}R}\right|\leq
 |R'|-\delta+1,
 \end{equation}
then we have
\begin{equation*}
\rank\left(\bigcup_{R\in \cV}R\right)\leq \left|\bigcup_{R\in \cV}R\right|-|\cV|(\delta-1).
\end{equation*}
\end{lemma}
\begin{proof}
Denote $|\cV|=\ell$ and $\cV=\{R_1,\dots,R_\ell\}\subseteq \cR$. For
each $R_{i}\in \cV$, \eqref{eqn_condition_joint} means that there
exists a $(\delta-1)$-subset $R'_{i}\subseteq R_{i}$ such that
$R'_{i}\cap(\bigcup_{j\in[\ell]\setminus\mathset{i}}R_j)=\emptyset$. Thus,
we can get $\ell$ pairwise disjoint subsets $R'_{1},R'_{2},\dots,
R'_{\ell}$.

By Definition \ref{def_r_delta_i}, $\rank(R_{i})=\rank(R_{i}\setminus
R'_{i})$ for $1\leq i\leq \ell$.  Therefore, we have
\begin{equation*}
\begin{split}
\rank\left(\bigcup_{R\in \cV}R\right)=\rank\left(\bigcup_{i\in[\ell]}(R_{i}\setminus R'_{i})\right)\leq \left|\bigcup_{i\in[\ell]}(R_{i}\setminus R'_{i})\right|&=\left|\bigcup_{R\in\cV}R\right|-\sum_{i\in[\ell]}|R'_{i}|\\
&=\left|\bigcup_{R\in \cV}R\right|-|\cV|(\delta-1).
\end{split}
\end{equation*}
\end{proof}

We note that when $\delta=2$, \eqref{eqn_condition_joint} is always
satisfied by the ECF $\cR$. We now continue with our exploration of
the properties of $\cR$.

\begin{lemma}\label{lemma_extend}
  Let $\cC$ be an $[n,k]_q$ linear code with all symbol
  $(r,\delta)$-locality. Let $\cR$ be the ECF given by Lemma
  \ref{lemma_initialization}.  If there are subsets $\cV\subseteq
  \cR'\subseteq \cR$ with $|\cV|\leq \lceil\frac{k}{r}\rceil-1$,
  $\rank(\bigcup_{R\in \cR'}R)=k$, and
  \begin{equation}\label{eqn_condition_joint_extend}
    \rank\left(\bigcup_{R\in \cV}R\right)\leq \left|\bigcup_{R\in \cV}R\right|-|\cV|(\delta-1)
  \end{equation}
  then we can obtain a $(\lceil\frac{k}{r}\rceil-1)$-set $\cV'$ with $\cV\subseteq\cV'\subseteq \cR'$ such that
  \begin{equation*}
    \rank\left(\bigcup_{R\in \cV'}R\right)\leq \left|\bigcup_{R\in \cV'}R\right|-|\cV'|(\delta-1).
  \end{equation*}
\end{lemma}
\begin{proof}
If $|\cV|=\lceil\frac{k}{r}\rceil-1$, then the lemma follows by
setting $\cV'=\cV$. Otherwise, we have
$|\cV|<\lceil\frac{k}{r}\rceil-1$. Since every $R\in\cR$ is an
$(r,\delta)$-repair set, $\rank(R)\leq r$. This means that
$\rank\left(\bigcup_{R\in\cV}R\right)<(\lceil\frac{k}{r}\rceil-1)r<k$.
Note that by the lemma requirements, $\rank\left(\bigcup_{R\in
  \cR'}R\right)=k$, which implies that there exists a $R'\in
\cR'\setminus \cV$ such that $\rank(R'\cup (\bigcup_{R\in
  \cV}R))>\rank(\bigcup_{R\in \cV}R)$.
We recall, however, that since $R'$ is an $(r,\delta)$-repair set, if $R^*\subseteq R'$, $\abs{R^*}\geq \abs{R'}-\delta+1$, then $\spn(R^*)=\spn(R')$. It follows that $R'$ cannot have a large intersection with $\bigcup_{R\in \cV}R$, namely,
$$\abs{R'\cap \left(\bigcup_{R\in \cV}R\right)}<|R'|-\delta+1.$$
Hence, there exists a $R''\subseteq R'\setminus \left(\bigcup_{R\in
  \cV}R\right)$ with $|R''|=\delta-1$. Again, using the fact that $R'$
is an $(r,\delta)$-repair set and $\abs{R'\setminus
  R''}=\abs{R'}-\delta+1$, we have $\rank(R')=\rank(R'\setminus R'')$, and
therefore,
\begin{equation*}
\begin{split}
\rank\left(\bigcup_{R\in \cV\cup\{R'\}}R\right)&=\rank\left(\left(\bigcup_{R\in \cV\cup\{R'\}}R\right)\setminus R''\right)\\
&\leq \left|R'\setminus \left(\left(\bigcup_{R\in \cV}R\right)\cup R''\right)\right|+\rank\left(\bigcup_{R\in \cV}R\right)\\
&\leq \left|R'\setminus \left(\bigcup_{R\in \cV}R\right)\right|-\delta+1+\left|\bigcup_{R\in \cV}R\right|-|\cV|(\delta-1)\\
&=\left|\bigcup_{R\in \cV\cup\{R'\}}R\right|-|\cV\cup\{R'\}|(\delta-1),
\end{split}
\end{equation*}
where the last inequality holds by the fact $R''\subseteq R'\setminus \left(\bigcup_{R\in \cV}R\right)$ and \eqref{eqn_condition_joint_extend}.
Therefore, repeating the above operations, we can extend $\cV$ to a $(\lceil\frac{k}{r}\rceil-1)$-subset $\cV'\subseteq \cR'$ such
that
$$\rank\left(\bigcup_{R\in \cV'}R\right)\leq \left|\bigcup_{R\in \cV'}R\right|-|\cV'|(\delta-1).$$
\end{proof}

\begin{lemma}\label{lemma_specail_cases}
  Let $\cC$ be an $[n,k]_q$ linear code with all symbol
  $(r,\delta)$-locality. Let $\cR$ be the ECF given by Lemma
  \ref{lemma_initialization}.  Assume $\cV\subseteq \cR$ such that
  $|\cV|\leq \lceil\frac{k}{r}\rceil-1$. If there exists a $R'\in \cV$
  such that
  \begin{equation}\label{eqn_condition_joint>r}
    \left|R'\bigcap \parenv{\bigcup_{R\in \cV\setminus\{R'\}}R}\right|> |R'|-\delta+1,
  \end{equation}
  then there exists $S\subseteq [n]$ with
  $\rank(S)=k-1$ and
\begin{equation*}
  |S|\geq
  k+\parenv{\left\lceil\frac{k}{r}\right\rceil-1}(\delta-1).
\end{equation*}
\end{lemma}
\begin{proof}
  Assume $\cV$ satisfies \eqref{eqn_condition_joint>r}. Let
  $\cV'\subseteq \cV$ be a minimal subset for which
  \eqref{eqn_condition_joint>r} holds, i.e., there exists a set $R'\in
  \cV'$ with $|R'\cap (\bigcup_{R\in \cV'\setminus
    \{R'\}}R)|>|R'|-\delta+1$, which in turn implies that
  $\spn(R')\subseteq \spn(\bigcup_{R\in \cV'\setminus \{R'\}}R)$. By
  the minimality of $\cV'$, the set $\cV'\setminus \{R'\}$ satisfies
  the requirements of Lemma \ref{lemma_for_rank_B_i}, which implies
  \begin{equation*}
    \rank\left(\bigcup_{R\in \cV'\setminus \{R'\}}R\right)\leq \left|\bigcup_{R\in \cV'\setminus \{R'\}}R\right|-|\cV'\setminus \{R'\}|(\delta-1).
  \end{equation*}

  As noted before, $\spn(R')\subseteq \spn(\bigcup_{R\in \cV'\setminus
    \{R'\}}R)$, and since trivially $\rank(\bigcup_{R\in \cR}R)=k$, we
  also necessarily have $\rank(\bigcup_{R\in
    \cR\setminus\{R'\}\}}R)=k$.  Therefore, by Lemma
  \ref{lemma_extend}, we can extend $\cV'\setminus \{R'\}$ to a
  $(\lceil\frac{k}{r}\rceil-1)$-subset $\cV''\subseteq\cR\setminus
  \{R'\}$ such that
  $$\rank\left(\bigcup_{R\in \cV''} R \right)\leq \left|\bigcup_{R\in \cV''}R\right|-|\cV''|(\delta-1)=\left|\bigcup_{R\in \cV''}R\right|-\left(\left\lceil\frac{k}{r}\right\rceil-1\right)(\delta-1).$$

  Considering the set $\cV''\cup \{R'\}$, we have
  \begin{equation}\label{eqn_V_2}
    \begin{split}
      \rank\left(\bigcup_{R\in \cV''\cup\{R'\}}R\right)=\rank\left(\bigcup_{R\in \cV''}R\right)&\leq \left|\bigcup_{R\in \cV''}R\right|-\left(\left\lceil\frac{k}{r}\right\rceil-1\right)(\delta-1)\\
      &\leq \left|\bigcup_{R\in \cV''\cup\{R'\}}R\right|-1-\left(\left\lceil\frac{k}{r}\right\rceil-1\right)(\delta-1),
    \end{split}
  \end{equation}
  where the last inequality holds due to the fact that
  $R'\not\subseteq \bigcup_{R\in \cV''}R$ by the properties of the ECF
  $\cR$.

  Since
  \[\rank\left(\bigcup_{R\in
    \cV''\cup\{R'\}}R\right)=\rank\left(\bigcup_{R\in
    \cV''}R\right)\leq
  \parenv{\left\lceil\frac{k}{r}\right\rceil-1}r\leq k-1,\] we can
  find a set $S$ with $\rank(S)=k-1$ by taking $\bigcup_{R\in
    \cV''\cup\{R'\}}R$ and adding arbitrary coordinates until
  reaching the desired rank. This set $S$ has size
  \[|S|\geq k-1-\rank\left(\bigcup_{R\in
    \cV''\cup\{R'\}}R\right)+\abs{\bigcup_{R\in \cV''\cup\{R'\}}R}\geq
  k+\parenv{\left\lceil\frac{k}{r}\right\rceil-1}(\delta-1),\]
  which follows from \eqref{eqn_V_2}.
\end{proof}

\textbf{\textit{ Proof of Lemma \ref{lemma_find_V_1}:}}
If the requirements of Lemma \ref{lemma_specail_cases} hold for
$\cV$, then the desired $S$ may be obtained by Lemma
\ref{lemma_specail_cases}, and we are done. Otherwise, $\cV$ does not satisfies
the requirements of Lemma \ref{lemma_specail_cases}, and then using
Lemmas \ref{lemma_for_rank_B_i} and \ref{lemma_extend} (setting
$\cR'=\cR$ in the latter), $\cV$ may be extended to a set
$\cV'\subseteq\cR$ with $\lceil \frac{k}{r}\rceil-1$ elements
satisfying
\begin{equation*}
\rank\left(\bigcup_{R\in \cV'}R\right)\leq \left|\bigcup_{R\in \cV'}R\right|-|\cV'|(\delta-1)=\left|\bigcup_{R\in \cV'}R\right|-\left(\left\lceil\frac{k}{r}\right\rceil-1\right)(\delta-1).
\end{equation*}
Recall that $k=ru+v$, with $0\leq v\leq r-1$. It now follows that
\begin{equation}\label{eqn_rank_S_2}
\begin{split}
k-1-\rank\left(\bigcup_{R\in \cV'}R\right)
&\geq ru+v-1-\left|\bigcup_{R\in \cV'}R\right|+|\cV'|(\delta-1)\\
&=
\begin{cases}
 u(r+\delta-1)-\left|\bigcup_{R\in \cV'}R\right|+v-1, &\text{ if $v\ne 0$},\\
 r+(u-1)(r+\delta-1)-\left|\bigcup_{R\in \cV'}R\right|+v-1, &\text{ if $v=0$},\\
\end{cases}\\
&=
\begin{cases}
 |\cV'|(r+\delta-1)-\left|\bigcup_{R\in \cV'}R\right|+v-1, &\text{ if } v\ne 0,\\
 r+|\cV'|(r+\delta-1)-\left|\bigcup_{R\in \cV'}R\right|-1, &\text{ if } v=0,\\
\end{cases}\\
&\overset{(a)}{\geq}\begin{cases}
 |\cV|(r+\delta-1)-\left|\bigcup_{R\in \cV}R\right|+v-1, &\text{ if } v\ne 0,\\
 r+|\cV|(r+\delta-1)-\left|\bigcup_{R\in \cV}R\right|-1, &\text{ if } v=0,\\
\end{cases}\\
&\overset{(b)}{\geq} \begin{cases}
 \Delta+v-1, &\text{ if } v\ne 0,\\
 r+\Delta-1, &\text{ if } v=0,\\
\end{cases}\\
\end{split}
\end{equation}
where $(a)$ follows from the fact that $|R|\leq r+\delta-1$ for all
$R\in \cV'$, and $(b)$ follows from \eqref{eqn_delta}.

For the case $v\ne 0$,
$\lceil\frac{k+\Delta}{r}\rceil=u+\lceil\frac{v+\Delta}{r}\rceil>\lceil
\frac{k}{r}\rceil=u+1$ means that $\Delta+v> r$, i.e., $\Delta+v-1\geq
r$. Thus, by \eqref{eqn_rank_S_2} and $\Delta>0$,
\begin{equation}\label{eqn_rank_V3}
\rank\left(\bigcup_{R\in \cV'}R\right)\leq k-1- r,
\end{equation}
for both $v=0$ and $v\ne 0$.

Again, by the same analysis as in Lemma \ref{lemma_extend}, we can
obtain yet another set $R'\in \cR\setminus \cV'$ with $\rank(R'\cup
(\bigcup_{R\in \cV'})R)>\rank(\bigcup_{R\in \cV'}R)$ and then
\begin{equation}\label{eqn_rank_V_3_R}
\rank\left(\bigcup_{R\in \cV'\cup \{R'\}}R\right)\leq \left|\bigcup_{R\in \cV'\cup \{R'\}}R\right|-\left|\cV'\cup \{R'\}\right|(\delta-1)=\left|\bigcup_{R\in \cV'\cup \{R'\}}R\right|-\left\lceil\frac{k}{r}\right\rceil(\delta-1).
\end{equation}
Note that $\rank(\bigcup_{R\in \cV'\cup \{R'\}}R)\leq
\rank(\bigcup_{R\in \cV'}R)+r\leq k-1$ by \eqref{eqn_rank_V3}.
Therefore, construct $S$ by adding coordinates to $\bigcup_{R\in
  \cV'\cup \{R'\}}R$ until reaching sufficient rank, $\rank(S)=k-1$, and then by
\eqref{eqn_rank_V_3_R} we have
\[|S|\geq k-1-\rank\left(\bigcup_{R\in \cV'\cup \{R'\}}R\right)+\left|\bigcup_{R\in \cV'\cup \{R'\}}R\right| \geq k-1+\left\lceil\frac{k}{r}\right\rceil(\delta-1)\geq k+\parenv{\left\lceil\frac{k}{r}\right\rceil-1}(\delta-1),\]
which completes the proof.\qed


\begin{thebibliography}{11}

\bibitem{ABHMT} A. Agarwal, A. Barg, S. Hu, A. Mazumdar, and I. Tamo, ``Combinatorial alphabet-dependent bounds for locally recoverable codes," \emph{IEEE Trans. Inf. Theory,}
vol. 64, no. 5, pp. 3481-3492, 2018.

\bibitem{BHP} R. C. Baker, G. Harman, and J. Pintz, ``The difference between consecutive primes, {II},'' \emph{Proceedings of the London Mathematical Society},
  vol.~83, no.~3, pp. 532--562, Nov. 2001.

\bibitem{BCGLP} A. Beemer, R. Coatney, V. Guruswami, H. H. L\'{o}pez, and  F. Pi\~{n}ero, ``Explicit optimal-length locally repairable codes of distance 5," arXiv:1810.03980, 2018.

\bibitem{BT} S. Bhadane and A. Thangaraj, ``Irregular recovery and unequal locality for locally recoverable codes with availability," arXiv:1705.05005, 2017.

\bibitem{CM} V. R. Cadambe and A. Mazumdar, ``Bounds on the size of locally recoverable codes," \emph{IEEE Trans. Inf. Theory,} vol. 61, no. 11, pp. 5787-5794, Nov. 2015.

\bibitem{CCFT} H. Cai, M. Cheng, C. Fan, and X. Tang, ``Optimal locally repairable systematic codes based on packings," \emph{IEEE Trans. Commun.}, vol. 67, no. 1, pp. 39-49, Jan. 2019.

\bibitem{CMST} H. Cai, Y. Miao, M. Schwartz, and X. Tang, ``On optimal locally repairable codes with multiple disjoint repair sets," submitted to \emph{IEEE Trans. Inf. Theory}.

\bibitem{CD} C. Colbourn and J. Dinitz, \emph{Handbook of Combinatorial Designs, Second Edition,} Chapman
\& Hall/CRC, 2006.


\bibitem{FY} M. Forbes and S. Yekhanin, ``On the locality of codeword symbols in
non-linear codes," \emph{Discr. Math.,} vol. 324, no. 6, pp. 78-84, Jun. 2014.

\bibitem{GHSY} P. Gopalan, C. Huang, H. Simitci, and S. Yekhanin, ``On the locality of codeword symbols," \emph{IEEE Trans. Inf. Theory,} vol. 58, no. 11, pp. 6925-6934, Nov. 2012.

\bibitem{GM} Y. Gu and Y. Miao, ``Union-intersection-bounded families and their applications," to appear in \emph{Discrete Applied Mathematics}.

\bibitem{GXY} V. Guruswami, C. Xing, and C. Yuan, ``How long can optimal locally repairable codes be?" arXiv:1807.01064, 2018.

\bibitem{HX} J. Hao and S. Xia, ``Constructions of optimal binary locally repairable codes with multiple repair groups," \emph{IEEE Commun. Lett.}, vol. 20, no. 6, pp. 1060-1063, Jun. 2016.

\bibitem{HCL} C. Huang, M. Chen, and J. Li, ``Pyramid codes: Flexible schemes to trade space for access efficiency in reliable data storage systems," In NCA, 2007.

\bibitem{HSXOCY} C. Huang, H. Simitci, Y. Xu, A. Ogus, B. Calder, P. Gopalan, J. Li, and S. Yekhanin, ``Erasure coding in windows azure storage," USENIX
Association, 2012.

\bibitem{Jin} L. Jin, ``Explicit construction of optimal locally recoverable codes of distance $5$ and $6$ via binary constant weight codes,"  arXiv:1808.04558, 2018.

\bibitem{J} S. M. Johnson, ``A new upper bound for error-correcting codes," \emph{IEEE Trans. Inf. Theory,} vol. 8, vol. 2, pp. 203-207, Apr. 1962.

\bibitem{K} P. Keevash, ``The existence of designs," arXiv:1401.3665v1, 2014.

\bibitem{KL} G. Kim and J. Lee, ``Locally repairable codes with unequal locality requirements," \emph{IEEE Trans. Inf. Theory,} vol. 64, no. 11, pp. 7137-7152, Nov. 2018.


\bibitem{KBTY} O. Kolosov, A. Barg, I. Tamo, and G. Yadgar, ``Optimal LRC codes for all lenghts $n\leq q$," arXiv:1802.00157, 2018.

\bibitem{LMC} J. Liu, S. Mesnager, and L. Chen, ``New constructions of optimal locally recoverable codes via good polynomials," \emph{IEEE Trans. Inf. Theory,} vol. 64, no. 2, pp. 889-899, Feb. 2018.

\bibitem{LXY} Y. Luo, C. Xing, and C. Yuan, ``Optimal locally repairable codes of distance 3 and 4 via cyclic codes,"  arXiv:1801.03623, 2018.

\bibitem{MS} F. J. MacWilliams and N. J. A. Sloane,  \emph{The Theory of Error-Correcting Codes}, North-Holland, 1977.

\bibitem{PHO} L. Pamies-Juarez, H. Hollmann, and F. Oggier, ``Locally repairable codes with multiple repair alternatives," In Proc. of IEEE ISIT, 2013.

\bibitem{PKLK} N. Prakash, G. M. Kamath, V. Lalitha, and P. V. Kumar, ``Optimal linear codes with a local-error-correction property," In Proc.
of IEEE ISIT, 2012.

\bibitem{RKSV} A. S. Rawat, O. O. Koyluoglu, N. Silberstein, and S. Vishwanath, ``Optimal locally repairable and secure codes for distributed storage systems," \emph{IEEE Trans. Inf. Theory,} vol. 60, no. 1, pp. 212-236, Jan. 2014.

\bibitem{RPDV} A. S. Rawat, D. S. Papailopoulos, A. G. Dimakis, and S. Vishwanath, ``Locality and availability in distributed storage," \emph{IEEE Trans. Inf. Theory,} vol. 62, no. 8, pp.
4481-4493, Aug. 2016.

\bibitem{SAK} B. Sasidharan, G. K. Agarwal, and P. V. Kumar, ``Codes with hierarchical locality," In Proc. of IEEE ISIT, 2015.

\bibitem{SES} N. Silberstein, T. Etzion, and M. Schwartz, ``Locality and availability of array codes constructed from subspaces," to appear in \emph{IEEE Trans. Inf. Theory}.

\bibitem{SDYL} W. Song, S. H. Dau, C. Yuen, and T. J. Li, ``Optimal locally  repairable linear codes," \emph{IEEE J. Slect. Areas Commun.,} vol. 32, no. 5, pp.
1019-1036, May 2014.

\bibitem{TPD} I. Tamo, D. Papailiopoulos, and A. Dimakis, ``Optimal locally repairable codes and connections to matroid theory," \emph{IEEE Trans. Inf. Theory,} vol. 62, no. 12, pp. 6661-6671, Dec. 2016.

\bibitem{TB} I. Tamo and A. Barg, ``A family of optimal locally recoverable codes,"
\emph{IEEE Trans. Inf. Theory,} vol. 60, no. 8, pp. 4661-4676, Aug. 2014.

\bibitem{WZ} A. Wang and Z. Zhang, ``Repair locality with multiple erasure tolerance," \emph{IEEE Trans. Inf. Theory,} vol. 60, no. 11, pp.
6979-6987, Nov. 2014.

\bibitem{WZ15} A. Wang and Z. Zhang, ``An integer programming-based bound for locally repairable codes," \emph{IEEE Trans. Inf. Theory,} vol. 61, no. 10, pp.
5280-5294, Oct. 2015.

\bibitem{WFEH} T. Westerb\"{a}ck, R. Freij-Hollanti, T. Ernvall, and C. Hollanti, ``On the combinatorics of locally repairable codes via matroid theory," \emph{IEEE Trans. Inf. Theory,} vol. 62, no. 10, pp. 5296-5315, Oct. 2016.

\end{thebibliography}
\end{document}